\documentclass[11pt, twoside]{article}
\usepackage{bcc-cambridge-class}
\usepackage{amsmath,amssymb,enumerate,mathtools}
\usepackage{tikz}
\usepackage{mathdots}
\usepackage{algorithm}
\usepackage{algpseudocode}
\usetikzlibrary{arrows}

\newcommand{\sources}{S^+}
\newcommand{\sinks}{S^-}
\newcommand{\terminals}{S}
\newcommand{\net}[1]{N^{#1}}
\newcommand{\oo}{o}
\newcommand{\cost}{c}
\newcommand{\capacity}{u}
\newcommand{\transit}{\tau}
\newcommand{\supersource}{\psi}
\newcommand{\timehorizon}{\theta}

\makeatother

\tikzstyle{node}=[circle, inner sep = 0pt, minimum size = 1em, fill, fill = black!80, text = white]
\tikzstyle{edge}=[very thick, draw = black!80]
\tikzstyle{arc}=[edge,->]

\bcctitle{An Introduction to Transshipments Over Time}
\bccshorttitle{An Introduction to Transshipments Over Time}

\bccname{Martin Skutella}
\bccshortname{M.~Skutella} 

\bccaddressa{TU Berlin\\ Institut für Mathematik, MA 5-2\\ Str.~des 17.~Juni 136\\ 10623 Berlin, Germany}
\bccemaila{martin.skutella@tu-berlin.de}
\begin{document}
%% BCC editors to set page numbers by uncommenting the line below and entering the correct number
%\setcounter{page}{<insert page number here>} 

\makebcctitle

%------------------------------------------------

\begin{abstract}
Network flows over time are a fascinating generalization of classical (static) network flows,
introducing an element of time. They naturally model problems where travel and transmission
are not instantaneous and flow may vary over time. Not surprisingly, flow over time problems
turn out to be more challenging to solve than their static counterparts. In this survey, we
mainly focus on the efficient computation of transshipments over time in networks with
several source and sink nodes with given supplies and demands, which is arguably the most
difficult flow over time problem that can still be solved in polynomial time.
\end{abstract}

%------------------------------------------------

\section{Introduction}
\label{sec:intro}

The study of flows over time goes back to the work of Ford and
Fulkerson~\cite{FordFulkerson1958} who introduced the topic under the name \defword{dynamic
flows} and devoted a section of their seminal book~\cite{FordFulkerson1962} to that topic. In
contrast to classical static flows, flows over time include an element of time. They model
the fluctuation of flow along arcs over time as well as non-instantaneous travel through a
network where flow traveling along an arc experiences a certain delay. Such effects play a
critical role in many network routing problems such as road, pedestrian, rail, or air traffic
control. Further applications include production systems, communication and data networks,
and financial flows.; see,
e.g.,~\cite{Aronson1989,KoehlerMoehringSkutella2009,PowellJailletOdoni1995}.

Flow over time problems turn out to considerably more difficult to solve than their static
counterparts. While Ford and Fulkerson~\cite{FordFulkerson1958} show how to reduce the
maximum flow over time problem to a static minimum cost flow problem, Klinz and
Woeginger~\cite{KlinzWoeginger2004} prove that finding minimum cost flows over time is
already an weakly NP-hard problem. Hall, Hippler, and Skutella~\cite{HallHipplerSkutella2007}
show the same for multicommodity flows over time, with some variants even being strongly
NP-hard. Hoppe and Tardos~\cite{HoppeTardos2000,Hoppe1995} study the transshipment over time
problem where flow needs to be sent from several source nodes with given supplies to several
sink nodes with given demands. For the case of static flows, finding a transshipment can be
reduced to finding a maximum flow from a super-source to a super-sink. Unfortunately, this
idea no longer works for the transshipment over time problem. Hoppe and Tardos present an
efficient algorithm that relies on a sequence of parametric submodular function minimizations
and Schlöter and Skutella~\cite{SchloeterSkutella2017,Schloeter2018} observe that a single
submodular function minimization is sufficient for finding a transshipment over time that
sends flow fractionally.

The main purpose of this paper is to provide a gentle introduction to the topic of
transshipments over time and related problems, that may also be used for the purpose of
teaching in an advanced course on combinatorial optimization. For a broader overview of flows
over time we refer to the survey papers by Aronson~\cite{Aronson1989}, Powell, Jaillet, and
Odoni~\cite{PowellJailletOdoni1995}, Kotnyek~\cite{Kotnyek2003}, Lovetskii and
Melamed~\cite{LovetskiiMelamed1987}, and Skutella~\cite{Skutella2009}, as well as the PhD
thesis of Schlöter~\cite{Schloeter2018} on the topic.

\subsubsection*{Outline}
In Section~\ref{sec:notation} we introduce notation and some basic definitions. This includes
a discussion of different time models for flows over time and, in particular, the role of
time-expanded networks. In contrast to all previous literature on the considered flow over
time problems, we manage to completely avoid the use of time-expanded networks in the
remainder of the paper. Section~\ref{sec:max-flow} is devoted to the algorithm of Ford and
Fulkerson and its analysis for computing maximum flows over time via temporally repeated
flows. This is an important building block for the more advanced flow over time problems
covered in later sections. In Section~\ref{sec:submodular} we give a novel proof that a
natural set function on subsets of terminals given by corresponding maximum flow over time
values is submodular. This fact turns out to be important for the efficient solution of the
transshipment over time problem. Section~\ref{sec:earliest-arrival-flows} is devoted to
earliest arrival flows which constitute a fascinating refinement of maximum flows over time.
We treat them here since earliest arrival flows require an interesting generalization of
temporally repeated flows, so-called chain-decomposable flows, which are needed in more
generality for transshipments over time. In Section~\ref{sec:lex-max-flows} we present Hoppe
and Tardos' algorithm for efficiently computing lexicographically maximum flows over time.
This algorithm is the most important subroutine for solving the transshipment over time
problem, which is the subject of the final Section~\ref{sec:transshipments}.

%------------------------------------------------

\section{Notation and preliminaries}
\label{sec:notation}

We consider a network~$\net{}$ given by a directed graph~$(V,A)$ with node set~$V$ and arc
set~$A$. Every arc~$a\in A$ has a positive \defword{capacity}~$u_a\in\R_{>0}\cup\{\infty\}$
and a non-negative \defword{transit time}~$\transit_a\in\R_{\geq0}$. We also refer to transit
times as \defword{length} and \defword{cost}. Network~$\net{}$ comes with a set of
\defword{terminals}~$S\subseteq V$ of cardinality~$k\coloneqq|S|$ that is partitioned into
two subsets of \defword{sources}~$\sources$ and \defword{sinks}~$\sinks$.

We refer to an arc~$a\in A$ with start node~$v$ and end node~$w$ simply as~$vw$. Without loss of
generality, there are no parallel or opposite arcs in network~$\net{}$, sources have no incoming,
and sinks no outgoing arcs. Having said that, for simplicity of notation, we include for each
arc~$vw$ its reverse arc~$wv$ in~$A$. 
The reverse arc's capacity is set to~$u_{wv}\coloneqq0$ and its transit time (or cost or length)
to~$\transit_{wv}\coloneqq-\transit_{vw}\leq0$. One advantage of this convention is that all arcs
of the residual network introduced below will be contained in~$A$.

For a path or cycle~$C$, we denote its total transit time (or length or cost)
by~$\transit(C)$. 
Whenever we use the term \defword{shortest path} or \defword{shortest cycle}, we mean a path/cycle of minimum transit time.
For simplicity only, we assume that the only zero length cycles in~$A$
consist of pairs of opposite arcs. This assumption is without loss of generality, because it
may be enforced by adding infinitesimal (and thus purely symbolic) values to the arc transit
times.

\subsubsection*{Static flows}
A flow (or \defword{static flow})~$x$ in network~$\net{}$ assigns a flow value~$x_{vw}$ to every
arc~$vw$ and the negative of this value~$x_{wv}=-x_{vw}$ to the opposite arc~$wv$. Intuitively,
the flow value~$x_{vw}$ is the amount of flow that is sent from node~$v$ to node~$w$ through
arc~$vw$. A flow~$x$ must satisfy \defword{flow conservation} at every non-terminal node~$v\in
V\setminus\terminals$, meaning that the net amount of flow arriving at~$v$ is zero. That
is,~$\sum_wx_{vw}=0$, where we use the convention that~$x_{vw}=0$ if~$vw$ is not an arc. The cost
of flow~$x$ is~$\cost(x)\coloneqq\sum_{\{v,w\}}\transit_{vw}x_{vw}$, again
setting~$\transit_{vw}x_{vw}=0$ if~$vw$ is not an arc.
Notice that the sum in the definition of~$\cost(x)$ goes over node pairs~$\{v,w\}$
with~$v\neq w$, taking only one arc~$vw$ of each pair of opposite arcs~$vw$ and~$wv$ into
account; since~$\transit_{wv}x_{wv}=(-\transit_{vw})(-x_{vw})=\transit_{vw}x_{vw}$, the
cost~$\cost(x)$ is well-defined.

A \defword{(static) circulation}~$x$ is a flow that satisfies flow conservation at every node of
the network. A flow or circulation~$x$ is \defword{feasible} if it satisfies the capacity
constraint~$x_{vw}\leq u_{vw}$ for every arc~$vw$.
For a feasible flow~$x$ in network~$\net{}$, the \defword{residual network}~$\net{}_x$ is
obtained from network~$\net{}$ by replacing the original arc capacity~$u_{vw}$ with the
\defword{residual capacity}~$u^x_{vw}\coloneqq u_{vw}-x_{vw}\geq0$, for every arc~$vw$. 
Due to our convention of reverse arcs~$wv$ with~$u_{wv}=0$ and~$x_{wv}=-x_{vw}$, this definition of residual capacities is in line with the standard notion of reverse arcs in the residual network, since~$u_{wv}^x = u_{wv} - x_{wv} = 0 - (-x_{vw})=x_{vw}$.
Notice
that if we take a feasible flow~$x$ in~$\net{}$ and a feasible flow~$y$ in the residual
network~$\net{}_x$ of~$x$, then their sum~$x+y$ is again a feasible flow in network~$\net{}$.

A \defword{chain flow}~$\gamma$ of value~$|\gamma|$ sends~$|\gamma|$ units of flow along a
cycle or a path~$C_\gamma$. The length (or transit time) of chain flow~$\gamma$
is~$\transit(\gamma)\coloneqq\transit(C_\gamma)$, its cost
is~$\cost(\gamma)=|\gamma|\transit(\gamma)$. If~$C_\gamma$ is a cycle, then~$\gamma$ is
obviously a circulation. Otherwise, if~$C_\gamma$ is a path,~$\gamma$ only violates flow
conservation at its two end nodes, which will always be terminals in this paper.

A multiset of chain flows~$\Gamma$ is a \defword{chain decomposition} of flow~$x$, if the sum
of all chain flows in~$\Gamma$ is equal to~$x$, and
thus~$c(x)=\sum_{\gamma\in\Gamma}\cost(\gamma)\eqqcolon\cost(\Gamma)$. The subset of chain
flows~$\gamma\in\Gamma$ such that~$C_\gamma$ contains arc~$vw$ is denoted by~$\Gamma_{vw}$. A
chain decomposition~$\Gamma$ of~$x$ is a \defword{standard chain decomposition} if all cycles
or paths~$C_\gamma$,~$\gamma\in\Gamma$, use arcs in the same direction as~$x$ does.

\subsubsection*{Flows over time}
First of all, for an intuitive understanding of flows over time, one can associate arcs of
the network with pipes in a pipeline system for transporting some kind of fluid. The length
of each pipeline determines the transit time of the corresponding arc (assuming that flow
progresses at a constant pace) while the width determines its capacity. For a comprehensive introduction to flows over time we refer to the survey~\cite{Skutella2009}.

Formally, a \defword{flow over time}~$f$ consists of a function~$f_{vw}:\R\to\R$ for every arc~$vw$. We
say that~$f_{vw}(\timehorizon')$ is the \defword{flow rate} (i.e., amount of flow per time
unit) entering arc~$vw$ at its tail node~$v$ at time~$\timehorizon'$. The flow particles
entering arc~$vw$ at time~$\timehorizon'$ arrive at the head node~$w$ exactly~$\transit_{vw}$
time units later at time~$\timehorizon'+\transit_{vw}$. In particular, the outflow rate at the
head node~$w$ at time~$\timehorizon'$ equals~$f_{vw}(\timehorizon'-\transit_{vw})$. Accordingly,
for pairs of opposite arcs~$vw$ and~$wv$ we get
\[
f_{wv}(\timehorizon')=-f_{vw}(\timehorizon'-\transit_{vw})
\qquad\text{for all~$\timehorizon'\in\R$.}
\]
A flow over time~$f$ must satisfy \defword{flow conservation} at every non-terminal node,
meaning that the net rate of flow arriving at~$v\in V\setminus\terminals$ is zero at all
times. That is,~$\sum_wf_{vw}(\timehorizon')=0$ for all~$\timehorizon'\in\R$. A flow over
time~$f$ is \defword{feasible} if it satisfies the capacity
constraint~$f_{vw}(\timehorizon')\leq u_{vw}$ for all~$\timehorizon'\in\R$ and for every
arc~$vw$. A flow over time~$f$ has \defword{time
horizon}~$\timehorizon\in\R_{\geq0}\cup\{\infty\}$, if
\[
f_{vw}(\timehorizon')=0
\qquad
\text{for all~$\timehorizon'\in(-\infty,0)\cup[\timehorizon,\infty)$ and for every arc~$vw$.}
\]
All flow over time problems considered in this paper ask for a flow over time with a given time horizon that we always denote by~$\timehorizon$.

In order to determine, for example, the amount of flow that a flow over time~$f$ sends
through arc~$vw$ within some finite time interval like~$[0,\timehorizon)$, one has to
integrate the function~$f_{vw}$ over that time interval. In order for this to be well
defined, it is usually assumed that all arc functions~$f_{vw}$ must be Lebesgue-integrable.
All flows over time considered in this paper, however, have piecewise constant flow rates on
arcs so that we can safely ignore such technicalities.

\subsubsection*{Other flow over time models}
We mention that several variations of the described flow over time model exist in the
literature. For example, we require a very strict notion of flow conservation here, where all
flow arriving at a non-terminal node must leave that node again instantaneously.
Alternatively, one could allow for some limited or unlimited amount of flow to be temporarily
stored at a node. For all flow over time problems considered in this paper, this more relaxed
version of flow conservation does not lead to better solutions, such that we stick to the
strict and simple flow conservation model. There exist other flow over time problems,
however, where the possibility to temporarily store flow at intermediate nodes does make a
difference. In this context, multi-commodity flows over time are a prominent example;
see~\cite{FleischerSkutella2007,GrossSkutella2015}.

Another, more fundamental distinction of flow over time models is the underlying time model.
We work with the so-called \defword{continuous time model} where a flow over time assigns a
flow rate~$f_{vw}(\timehorizon')$ to every point in time~$\timehorizon'\in\R$. The original
definition by Ford and Fulkerson~\cite{FordFulkerson1958,FordFulkerson1962} is based on a
\defword{discrete time model}, where transit times~$\transit_{vw}$ and time
horizon~$\timehorizon$ are integral, and the flow along an arc~$vw$ is described by a
function~$g_{vw}:\Z\to\R$ where~$g_{vw}(\timehorizon')$ denotes the amount of flow sent at
time~$\timehorizon'$ into arc~$vw$ and arriving at the head node~$w$ at
time~$\timehorizon'+\transit_{vw}$. Fleischer and Tardos~\cite{FleischerTardos1998} point out
a strong connection between the two models. They show that many results and algorithms which
have been developed for the discrete time model can be carried over to the continuous time
model.

\subsubsection*{Time-expanded networks}
A distinctive feature of the discrete time model is the existence of \defword{time-expanded
networks} with time layers~$\timehorizon'=1,\dots,\timehorizon$ for time
horizon~$\timehorizon$, where every layer~$\timehorizon'$ contains a copy~$v_{\timehorizon'}$
of every node~$v$ of the underlying flow network. Moreover, for every arc~$vw$ there is a
copy~$v_{\timehorizon'}w_{\timehorizon'+\transit_{vw}}$ in the time-expanded network,
for~$\timehorizon'=1,\dots,\timehorizon-\transit_{vw}$. Every discrete flow over time with time
horizon~$\timehorizon$ thus corresponds to a static flow in the time-expanded network and
vice versa. As a consequence, flow over time problems in the discrete time model may be
solved as static flow problems in the time-expanded network. The downside of this approach,
however, is the huge size of the time-expanded network which is linear in the given time
horizon~$\timehorizon$ and therefore only pseudo-polynomial in the input size. Therefore,
time-expanded networks do not lead to efficient algorithms whose running time is polynomially
bounded in the input size.

Despite this deficit, time-expanded networks can sometimes be used in the analysis of
efficient algorithms; see, e.g., the work of Hoppe and Tardos~\cite{HoppeTardos2000} on
transshipments over time. Due to the strong connections between the discrete and the
continuous time model mentioned above, time-expanded networks can also play a role in the
analysis of algorithms for the latter time model; see, e.g.,~\cite{FleischerTardos1998}. This
comes at a price, however, since in the first instance the results obtained in this way are
usually limited to networks with integral or rational transit times.

One purpose of our paper is to avoid the use of time-expanded networks altogether. On the one
hand, this generalizes previously known results to networks with arbitrary (rational or
irrational) transit times. On the other hand, it sometimes even leads to arguably simpler
proofs.

%------------------------------------------------

\section{Maximum Flows Over Time}
\label{sec:max-flow}

Ford and Fulkerson~\cite{FordFulkerson1958} introduced the \defword{maximum flow over time
problem}. Here the task is to send as much flow as possible from the sources~$\sources$ to
the sinks~$\sinks$ within a given time bound~$\timehorizon$. This problem can be reduced to a
static minimum cost circulation problem in an extended network~$\net{\sources}$. We describe
the construction of this network in more generality since it will be used for more complex
flow over time problems in later sections as well.

For a subset of terminals~$X\subseteq\terminals$, the \defword{extended network}~$\net{X}$ is
constructed by adding super-source~$\supersource$ together with arcs~$(\supersource,s)$ for
all sources~$s\in\sources\cap X$ of infinite capacity~$\capacity_{(\supersource,s)}=\infty$
and zero transit time~$\transit_{(\supersource,s)}=0$, as well as arcs~$(t,\supersource)$ for
all sinks~$t\in\sinks\setminus X$ of infinite capacity~$\capacity_{(t,\supersource)}=\infty$
and negative transit time~$\transit_{(t,\supersource)}=-\timehorizon$; see
Figure~\ref{fig:extended-network}. 
Network~$\net{X}$ is designed to send flow from~$\supersource$ via the sources
in~$\sources\cap X$ through the original network~$\net{}$ into the sinks in~$\sinks\setminus
X$, and then back to~$\supersource$. Since a maximum flow over time may use all sources and
sinks, in the following we work with network~$\net{\sources}$.
\begin{figure}[h]
\centering
\begin{tikzpicture}
% sources
\fill [black!15] (-2,0) ellipse (0.6cm and 1.8cm);
\fill [black!30] (-1.4,0) arc (0:180:0.6cm and 1.8cm) -- cycle;
\node [node] (s1) at (-2,1.5) {};
\draw [arc] (s1) -- +(0.8,0.3);
\draw [arc] (s1) -- +(0.8,-0.3);
\node at (-2,0.9) {$\vdots$};
\node [node] (s2) at (-2,0.3) {};
\draw [arc] (s2) -- +(0.8,0.4);
\draw [arc] (s2) -- +(0.8,-0.1);
\node [node] (s3) at (-2,-0.3) {};
\draw [arc] (s3) -- +(0.8,0.1);
\draw [arc] (s3) -- +(0.8,-0.4);
\node at (-2,-0.9) {$\vdots$};
\node [node] (s4) at (-2,-1.5) {};
\draw [arc] (s4) -- +(0.8,0.3);
\draw [arc] (s4) -- +(0.8,-0.3);
\node at (-2.9,0) {$\sources$};

% sinks
\fill [black!15] (2,0) ellipse (0.6cm and 1.8cm);	
\fill [black!30] (1.4,0) arc (180:360:0.6cm and 1.8cm) -- cycle;
\node [node] (t1) at (2,1.5) {};
\draw [arc,<-] (t1) -- +(-0.8,0.3);
\draw [arc,<-] (t1) -- +(-0.8,-0.3);
\node at (2,0.9) {$\vdots$};
\node [node] (t2) at (2,0.3) {};
\draw [arc,<-] (t2) -- +(-0.8,0.4);
\draw [arc,<-] (t2) -- +(-0.8,-0.1);
\node [node] (t3) at (2,-0.3) {};
\draw [arc,<-] (t3) -- +(-0.8,0.1);
\draw [arc,<-] (t3) -- +(-0.8,-0.4);
\node at (2,-0.9) {$\vdots$};
\node [node] (t4) at (2,-1.5) {};
\draw [arc,<-] (t4) -- +(-0.8,0.3);
\draw [arc,<-] (t4) -- +(-0.8,-0.3);
\node at (2.9,0) {$\sinks$};

\node [node] (super) at (0,2.8) {\footnotesize$\supersource$};
\draw [arc] (super.200) .. controls (-3,2.3) and (-3,1.7) .. (s1);
\draw [arc] (super.170) .. controls (-3.8,2.5) and (-3.8,1.2) .. (s2);
\node at (-2.9,2.6) {$\transit_{\supersource r}=0$};
\node at (-2.02,1.15) {\footnotesize$\sources\!\!\cap\! X$};
\node at (-2.02,-0.68) {\footnotesize$\sources\!\!\setminus\!X$};
\node at (-2.8,1.5) {$\iddots$};

\draw [arc,<-] (super.-25) .. controls (3,2.3) and (3,1.7) .. (t1);
\draw [arc,<-] (super.20) .. controls (3.8,2.5) and (3.8,1.2) .. (t2);
\node at (3,2.6) {$\transit_{r\supersource}=-\timehorizon$};
\node at (1.98,-0.65) {\footnotesize$\sinks\!\!\cap\! X$};
\node at (1.98,1.15) {\footnotesize$\sinks\!\!\setminus\!X$};
\node at (2.8,1.5) {$\ddots$};

\end{tikzpicture}
\caption{Extended network~$\net{X}$ for a subset of terminals~$X\subseteq\terminals$}
\label{fig:extended-network}
\end{figure}
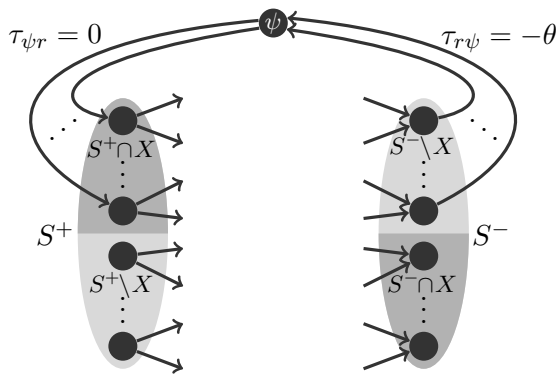

\begin{rem}
The problem originally studied by Ford and Fulkerson~\cite{FordFulkerson1958} is to send the
maximum possible amount of flow from a single source~$s$ to a single sink~$t$ within time
horizon~$\timehorizon$. Notice that this problem is equivalent to our generalized variant of
the problem described above: We could simply split up the super-source~$\supersource$ in
network~$\net{\sources}$ into two nodes~$s$ and~$t$ where~$s$ is connected to every source
in~$\sources$ and every sink in~$\sinks$ is connected to~$t$. Then, sending as
much flow as possible from~$s$ to~$t$ in this network is equivalent to sending as much flow
as possible from~$\sources$ to~$\sinks$ in~$\net{}$. In the following we thus stick to our
variant of the problem.	
\end{rem}

Ford and Fulkerson's algorithm starts by computing a standard chain decomposition~$\Gamma$ of
a minimum cost circulation~$x$ in~$\net{\sources}$. Each chain flow~$\gamma\in\Gamma$
sends~$|\gamma|$ units of flow along a cycle~$C_\gamma$ of nonpositive
cost~$\transit(C_\gamma)\leq0$ that contains the super-source~$\supersource$. If we
delete~$\supersource$ and its incident arcs from~$C_\gamma$, what remains is a source sink
path~$P_\gamma$ in~$\net{}$
with~$\transit(P_\gamma)=\transit(C_\gamma)+\timehorizon\leq\timehorizon$.

For each~$\gamma\in\Gamma$, Ford and Fulkerson's resulting flow over
time~$[\Gamma]^\timehorizon$ sends flow at rate~$|\gamma|$ into~$\gamma$'s source sink
path~$P_\gamma$ during the time interval~$\bigl[0,\timehorizon-\transit(P_\gamma)\bigr)$.
Notice that~$\transit(P_\gamma)\geq0$
and~$\timehorizon-\transit(P_\gamma)=-\transit(C_\gamma)\geq0$. Then, for
each~$\gamma\in\Gamma$, flow arrives at the final sink node of~$P$ at rate~$|\gamma|$ during
the time interval~$\bigl[\transit(P_\gamma),\timehorizon\bigr)$ and no flow remains in the
network after time~$\timehorizon$. Due to its special structure, the computed flow over
time~$[\Gamma]^\timehorizon$ is called a \defword{temporally repeated flow} or a
\defword{standard chain-decomposable flow}. 

\begin{Def}
\label{def:temporally-repeated}
Let~$\Gamma$ be a standard chain decomposition of a static circulation in~$\net{\sources}$
such that, for each~$\gamma\in\Gamma$,~$C_\gamma$ contains arc~$\supersource r_\gamma$ for
some source~$r_\gamma\in\sources$. Let~$P_\gamma$ denote the source sink path starting at~$r_\gamma$ that
is obtained from~$C_\gamma$ by deleting~$\supersource$ and its incident arcs. For
each~$\gamma\in\Gamma$, the \defword{standard chain-decomposable} (or \defword{temporally repeated}) \defword{flow~$[\Gamma]^\timehorizon$ with time
horizon~$\timehorizon$} sends flow at rate~$|\gamma|$ into path~$P_\gamma$ during the time
interval~$\bigl[0,\timehorizon-\transit(P_\gamma)\bigr)$.
\end{Def}

Ford and Fulkerson's algorithm is summarized in
Algorithm~\ref{alg:FordFulkerson}.
\begin{algorithm}
\caption{Ford and Fulkerson's maximum flow over time algorithm}
\label{alg:FordFulkerson}
\begin{algorithmic}[1]
\Require network~$\net{}$ and time horizon~$\timehorizon$
\Ensure maximum flow over time with time horizon~$\timehorizon$
\State compute a static minimum cost circulation~$x$ in~$\net{\sources}$
\State find a standard chain decomposition~$\Gamma$ of~$x$
\State \Return~$[\Gamma]^\timehorizon$
\end{algorithmic}
\end{algorithm}

By construction of the computed temporally repeated flow~$[\Gamma]^\timehorizon$, the flow
rate entering any arc~$a\in A$ is always non-negative and bounded
by~$\sum_{a\in\gamma}|\gamma|=x_a\leq u_a$. Moreover, it satisfies flow conservation
constraints at all intermediate nodes of the paths~$P_\gamma$,~$\gamma\in\Gamma$, since all
flow arriving at an intermediate node~$v$ while traveling along some path~$P_\gamma$
immediately continues its journey on the next arc leaving~$v$. Thus,~$[\Gamma]^\timehorizon$
is a feasible flow over time with time horizon~$\timehorizon$. The total amount of flow sent
from the sources to the sinks is
\begin{align}
\bigl|[\Gamma]^\timehorizon\bigr|
=\sum_{\gamma\in\Gamma}|\gamma|\bigl(\timehorizon-\transit(P_\gamma)\bigr)
=-\sum_{\gamma\in\Gamma}|\gamma|\transit(C_\gamma) =-\cost(x).
\label{eq:max-flow-value}
\end{align}
Notice that this flow value does not depend on the particular chain decomposition~$\Gamma$ but
only on the minimum cost~$\cost(x)$. In view of~\eqref{eq:max-flow-value}, the following
theorem states that Algorithm~\ref{alg:FordFulkerson} returns a maximum flow over time
in~$\net{}$ with time horizon~$\timehorizon$.

\begin{theorem}[Ford, Fulkerson~\cite{FordFulkerson1958};
Anderson, Philpott~\cite{AndersonPhilpott1994}]
\label{thm:max-flow-over-time}
The maximum
\linebreak 
amount of flow that a feasible flow over time with time horizon~$\timehorizon$ can send
from~$\sources$ to~$\sinks$ is equal to~$-\cost(x)$ where~$x$ is a static minimum cost
circulation in network~$\net{\sources}$.
\end{theorem}

While Ford and Fulkerson~\cite{FordFulkerson1958} proved this result in the discrete time
model via minimum cuts in time-expanded networks, the proof of Anderson and
Philpott~\cite{AndersonPhilpott1994} uses a more general concept of \defword{cuts over time}
introduced for the special case of zero transit times by Anderson, Nash, and
Philpott~\cite{AndersonNashPhilpott1982} (see also~\cite{AndersonNash1987}) and generalized to the
setting with non-zero transit times by Philpott~\cite{Philpott1990}; see also Fleischer and
Tardos~\cite{FleischerTardos1998}.

\begin{Def}
\label{def:cut-over-time} 
A \defword{cut over time} with time horizon~$\timehorizon$ in~$\net{}$ is given by non-negative
values~$(\alpha_v)_{v\in V}$, with~$\alpha_s\leq0$ for~$s\in\sources$
and~$\alpha_t\geq\timehorizon$ for~$t\in\sinks$. Its capacity is
\[
\sum_{vw\in A}\max\{0,\alpha_w-\transit_{vw}-\alpha_v\}u_{vw}.
\]
\end{Def}

The intuition behind Definition~\ref{def:cut-over-time} is as follows: a node~$v\in V$ is on
the \defword{source side} of the cut over time from time~$\alpha_v$ on. Before time~$\alpha_v$,
it is on the \defword{sink side}. In particular, within the time interval~$[0,\timehorizon)$,
sources in~$\sources$ are always on the source side and sinks in~$\sinks$ are always on the
sink side. Therefore, any flow particle traveling from a source to a sink during that time
interval must eventually cross the cut over time from the source to the sink side. The only
possibility to do so is to enter an arc~$vw$ while~$v$ is already on the source side, that is,
after time~$\alpha_v$, and to arrive at~$w$ while~$w$ is still on the sink side, that is,
before time~$\alpha_w$. Since it takes~$\transit_{vw}$ time to traverse arc~$vw$, the critical
time interval to enter is~$[\alpha_v,\alpha_w-\transit_{vw})$. As the length of the critical time
interval is~$\max\{0,\alpha_w-\transit_{vw}-\alpha_v\}$ and the flow rate entering arc~$vw$ is
bounded by its capacity~$u_{vw}$, the total amount of flow that can cross the cut over time
from the source to the sink side is bounded by its capacity.

\begin{theorem}[Anderson et al.~\cite{AndersonNashPhilpott1982}, Philpott~\cite{Philpott1990}]
\label{thm:cut-over-time-bound}
For a given network~$\net{}$, the capacity of a cut over time with time horizon~$\timehorizon$
is an upper bound on the maximum amount of flow that can be sent from~$\sources$ to~$\sinks$
by a feasible flow over time with time horizon~$\timehorizon$.
\end{theorem}

With this result at hand, it is not difficult to prove the optimality of the flow over time
computed by Algorithm~\ref{alg:FordFulkerson}.

\begin{proof}[of Theorem~\ref{thm:max-flow-over-time}]
According to~\eqref{eq:max-flow-value},~$-\cost(x)$ is a lower bound on the maximum flow
value. It thus remains to show that it is also an upper bound. This can be done via a suitable
cut over time. The value~$-\cost(x)$ of a static minimum cost circulation~$x$ in
network~$\net{\sources}$ is the optimum solution value of the following linear programming~(LP)
formulation (notice that the LP only uses variables~$x_a$ corresponding to the original (i.e.,~positive capacity) arcs~$a$ in~$\net{}$:
\begin{align*}
\max~ & -\sum_{a}\transit_a x_a\\
\text{s.t.}\quad 
& \sum_{a\in\delta^+(v)}x_a -\sum_{a\in\delta^-(v)}x_a = 0 
&& \text{for all~$v$,}\\
& x_a \leq u_a && \text{for all~$a$,}\\
& x_a \geq 0 && \text{for all $a$.}
\end{align*} 
To find a suitable cut over time, we consider the corresponding dual linear program: 
\begin{align}
\min~ 
& \sum_{a} u_a y_a
\notag\\
\text{s.t.}\quad 
& y_{a} + \alpha_v-\alpha_w\geq -\transit_{a} 
&& \text{for all~$a=vw$,}
\label{eq:dual-constraint}\\
& y_a\geq 0 
&& \text{for all~$a$.}
\notag
\end{align}
For an optimum dual solution it holds that $y_a=\max\{0,\alpha_w-\transit_a-\alpha_v\}$ for
all~$a=vw$. It only remains to show that the optimal dual values~$(\alpha_v)$ constitute a cut
over time as in Definition~\ref{def:cut-over-time}, whose capacity is then equal to the
optimum value~$-\cost(x)$ of the primal and dual linear program.

Notice that increasing or decreasing all dual variables~$\alpha_v$ by the same amount affects
neither feasibility nor the dual solution value. We may therefore assume
that~$\alpha_\supersource=0$. Since all arcs~$\supersource s$, $s\in\sources$,
and~$t\supersource$, $t\in\sinks$, have infinite capacity, complementary slackness implies
that their dual variables~$y_{\supersource s}$ and~$y_{t\supersource}$ are zero. Thus, the
dual constraints~\eqref{eq:dual-constraint} imply~$\alpha_s=\alpha_s-\transit_{\supersource
s}\leq\alpha_\supersource=0$ for all~$s\in\sources$. Similarly, for~$t\in\sinks$ the dual
constraint~\eqref{eq:dual-constraint}
implies~$\alpha_t\geq\alpha_\supersource-\transit_{t\supersource}=\timehorizon$.
\end{proof}

%------------------------------------------------

\section{Submodularity of the maximum flow over time function}
\label{sec:submodular}

In the previous section we have studied the problem of finding in the given network~$\net{}$ a
flow over time with time horizon~$\timehorizon$ that maximizes the amount of flow sent from
the sources~$\sources$ to the sinks~$\sinks$. In the context of transshipments over time, we
are also interested in how much flow can be sent from a particular subset of sources to a
particular subset of sinks within time~$\timehorizon$. More precisely, for some subset of
terminals~$X\subseteq\terminals$, we denote by~$\oo(X)$ the maximum amount of flow that can be
sent from~$\sources\cap X$ to~$\sinks\setminus X$ within time horizon~$\timehorizon$. As a
consequence of Theorem~\ref{thm:max-flow-over-time}, the value~$\oo(X)$ can be obtained via
one static minimum cost circulation computation in network~$\net{X}$.

\begin{cor}
\label{cor:max-flow-over-time}
For~$X\subseteq\terminals$, the maximum amount of flow~$\oo(X)$ that a feasible flow over time
with time horizon~$\timehorizon$ can send from~$\sources\cap X$ to~$\sinks\setminus X$ is
equal to~$-\cost(x)$ where~$x$ is a static minimum cost circulation in network~$\net{X}$.
\end{cor}

For the discrete time model, the values~$\oo(X)$ are minimum cut capacities in the
time-expanded network. In this model, the corresponding \defword{maximum flow over time
function}~$\oo:2^{\terminals}\to\R$ is thus a cut function, and as such
submodular; see Megiddo~\cite{Megiddo1974}. For the continuous time model with arbitrary (irrational)
transit times and time horizon, however, this line of argument cannot be directly applied and
we therefore present a different kind of proof.

\begin{theorem}
\label{thm:submodular}
The function~$\oo:2^{\terminals}\to\R$ is submodular.
\end{theorem}

\begin{proof}
For any~$X\subseteq\terminals$ and terminals~$r_1,r_2\in\terminals$ with~$r_1\neq r_2$ we need to argue that
\[
\oo(X\cup\{r_1,r_2\})-\oo(X\cup\{r_1\})\leq\oo(X\cup\{r_2\})-\oo(X).
\]
We distinguish four cases, depending on whether~$r_1$ and~$r_2$ are sources or sinks.

\emph{First case:} $r_1,r_2\in\sources$. Let~$x$ be a minimum cost circulation in
network~$\net{X}$ and~$y$ a minimum cost circulation in~$\net{X\cup\{r_1\}}_x$. Notice
that~$\net{X\cup\{r_1\}}$ compared to~$\net{X}$ has the additional arc~$\supersource r_1$
which might create negative cycles in the residual network~$\net{X\cup\{r_1\}}_x$ that must be
canceled by~$y$. Then,~$x+y$ is a minimum cost circulation in~$\net{X\cup\{r_1\}}$ and thus
\[
\oo(X\cup\{r_1\})=-\cost(x+y)=-\cost(x)-\cost(y)
\]
by linearity of the cost function. Moreover, let~$z$ be a minimum cost circulation
in~$\net{X\cup\{r_1,r_2\}}_{x+y}$. Again,~$\net{X\cup\{r_1,r_2\}}$ compared
to~$\net{X\cup\{r_1\}}$ has the additional arc~$\supersource r_2$ which might create negative
cycles in the residual network~$\net{X\cup\{r_1,r_2\}}_{x+y}$ that must be canceled by~$z$.
Then,~$x+y+z$ is a minimum cost circulation in~$\net{X\cup\{r_1,r_2\}}$, and therefore
\[
\oo(X\cup\{r_1,r_2\})=-\cost(x+y+z)=-\cost(x)-\cost(y)-\cost(z),
\]
again by linearity of the cost function. Furthermore, by construction,~$y+z$ is a feasible
circulation in~$\net{X\cup\{r_1,r_2\}}_{x}$ which can be decomposed into flow along cycles
in~$\net{X\cup\{r_1,r_2\}}_{x}$. The circulation given by the flow along cycles containing
arc~$\supersource r_1$ is denoted by~$y'$, the circulation given by flow along all remaining
cycles is denoted by~$z'$. Since circulation~$y'$ does not send any flow along
arc~$\supersource r_2$ in~$\net{X\cup\{r_1,r_2\}}_{x}$, it can be interpreted as a
feasible circulation in~$\net{X\cup\{r_1\}}_{x}$. Thus, since~$y$ is a minimum cost circulation
in~$\net{X\cup\{r_1\}}_{x}$, we get~$\cost(y')\geq\cost(y)$. This inequality then
implies~$\cost(z')\leq\cost(z)$, since~$y+z=y'+z'$, and thus
$\cost(y)+\cost(z)=\cost(y')+\cost(z')$. Moreover, by choice of~$y'$, circulation~$z'$ does
not send any flow along arc~$\supersource r_1$ in~$\net{X\cup\{r_1,r_2\}}_{x}$ and can thus
be interpreted as a feasible circulation in~$\net{X\cup\{r_2\}}_x$. In particular,~$\cost(z')$
is an upper bound on the cost of a minimum cost circulation in~$\net{X\cup\{r_2\}}_x$, that is,
\[
-\cost(z')\leq\oo(X\cup\{r_2\})-\oo(X).
\]
Putting everything together yields
\[
\oo(X\cup\{r_1,r_2\})-\oo(X\cup\{r_1\})=-\cost(z)\leq-\cost(z')\leq\oo(X\cup\{r_2\})-\oo(X).
\]

\emph{Second case:} $r_1,r_2\in\sinks$. We may use symmetric arguments as in the first case
above.

\emph{Third case:} $t\coloneqq r_1\in\sinks,s\coloneqq r_2\in\sources$. Let~$x$ be a
minimum cost circulation in~$\net{X}$ and~$x+y$ a minimum cost circulation in~$\net{X\cup\{t\}}$.
Notice that network~$\net{X\cup\{t\}}$ compared to~$\net{X}$ is missing arc~$t\supersource$.
Therefore, circulation~$y$ must cancel out the entire flow that~$x$ sends through that arc.
That is,~$y$ is a feasible circulation in the residual network~$\net{X}_x$
with~$y_{\supersource t}=x_{t\supersource}$ of minimum cost. In
particular,~$(x+y)_{t\supersource}=0$ such that~$x+y$ can indeed be interpreted as a
feasible circulation in~$\net{X\cup\{t\}}$ and therefore also in network~$\net{X\cup\{t,s\}}$,
which in addition contains arc~$\supersource s$.

Let~$z$ be a minimum cost circulation in~$\net{X\cup\{t,s\}}_{x+y}$, such that~$x+y+z$ is a
minimum cost circulation in~$\net{X\cup\{t,s\}}$, and therefore
\[
\oo(X\cup\{t,s\})=-\cost(x+y+z)=-\cost(x)-\cost(y)-\cost(z)
\]
by linearity of the cost function. By construction,~$y+z$ is a feasible circulation
in~$\net{X\cup\{s\}}_{x}$ which can be decomposed into flow along cycles
in~$\net{X\cup\{s\}}_{x}$. The circulation given by the flow along cycles containing
arc~$\supersource t$ is denoted by~$y'$, the circulation given by flow along all remaining
cycles is denoted by~$z'$. Since circulation~$y'$ does not send any flow along
arc~$\supersource s$ in~$\net{X\cup\{s\}}_{x}$, it can be interpreted as a feasible
circulation in~$\net{X}_{x}$. Moreover, $y'_{\supersource t}=x_{t\supersource}$, and
since~$y$ is a minimum cost circulation with this property in~$\net{X}_{x}$, we
get~$\cost(y')\geq\cost(y)$. This implies~$\cost(z')\leq\cost(z)$, since~$y+z=y'+z'$ and thus
$\cost(y)+\cost(z)=\cost(y')+\cost(z')$. Moreover,~$\cost(z')$ is an upper bound on the cost
of a minimum cost circulation in~$\net{X\cup\{s\}}_x$, that is,
\[
-\cost(z')\leq\oo(X\cup\{s\})-\oo(X).
\]
Putting everything together yields
\[
\oo(X\cup\{t,s\})-\oo(X\cup\{t\})=-\cost(z)\leq-\cost(z')\leq\oo(X\cup\{s\})-\oo(X).
\]

\emph{Fourth case:} $r_1\in\sources,r_2\in\sinks$. We may again use symmetric arguments as in
the third case above.
\end{proof}

As discussed in the beginning of this section, for the case of integral (or rational) transit
times, the maximum flow over time function~$\oo:2^{\terminals}\to\R$ is a cut function in a
time-expanded network. Even for the case of arbitrary (irrational) transit times, however, one
can argue that~$\oo:2^{\terminals}\to\R$ is the cut function of a more elaborate time-expanded
network~\cite{Manzhulina2023}.

%------------------------------------------------

\section{Earliest arrival flows}
\label{sec:earliest-arrival-flows}

In Section~\ref{sec:max-flow} we showed that there is always a maximum flow over time that
belongs to the structurally simple class of temporally repeated (or standard chain
decomposable) flows. Unfortunately, other flow over time problems, including the
transshipment over time problem, generally only allow for solutions that feature a more
complex temporal structure and therefore require a richer class of flows over time. An
illustrative example are earliest arrival flows, which are a fascinating refinement of
maximum flows over time.

\begin{Def}
An \defword{earliest arrival flow} in~$\net{}$ is a feasible flow over time maximizing the
amount of flow that has arrived at the sinks~$\sinks$ at each point in
time~$\timehorizon'\in[0,\timehorizon]$ simultaneously.
\end{Def}

The existence of earliest arrival flows has been first established by Gale~\cite{Gale1959} for
the discrete time model and by Philpott~\cite{Philpott1990} for the continuous time model. In
general, however, there is no temporally repeated earliest arrival flow; a small
counterexample is given in Figure~\ref{fig:counterexample}.
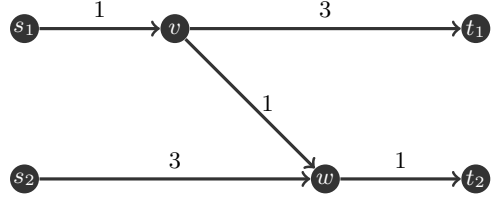
\begin{figure}[t]
\centering
\begin{tikzpicture}
\node [node] (s1) at (-3,1) {\footnotesize$s_1$};
\node [node] (s2) at (-3,-1) {\footnotesize$s_2$};
\node [node] (v) at (-1,1) {\footnotesize$v$};
\node [node] (w) at (1,-1) {\footnotesize$w$};
\node [node] (t1) at (3,1) {\footnotesize$t_1$};
\node [node] (t2) at (3,-1) {\footnotesize$t_2$};
\draw [arc] (s1) to node [above] {\footnotesize$1$} (v);
\draw [arc] (s2) to node [above] {\footnotesize$3$} (w);
\draw [arc] (v) to node [right] {\footnotesize$1$} (w);
\draw [arc] (v) to node [above] {\footnotesize$3$} (t1);
\draw [arc] (w) to node [above] {\footnotesize$1$} (t2);
\end{tikzpicture}
\caption{%
Network with two sources~$\sources=\{s_1,s_2\}$, two sinks~$\sinks=\{t_1,t_2\}$, and unit arc
capacities; numbers at arcs indicate transit times; a maximum flow over time with time
horizon~$\timehorizon=4$ has value~$1$ and must send flow at rate~$1$ into path~$s_1,v,w,t_2$
during time interval~$[0,1)$; a maximum temporally repeated flow with time
horizon~$\timehorizon=6$, however, has value~$4$ and must send flow at rate~$1$ into both
paths~$s_1,v,t_1$ and~$s_2,w,t_2$ during time interval~$[0,2)$; notice that there is no
temporally repeated flow that is maximum for~$\timehorizon=4$ and~$\timehorizon=6$
simultaneously.
}
\label{fig:counterexample}
\end{figure}

Wilkinson~\cite{Wilkinson1971} and Minieka~\cite{Minieka1973} present algorithms for
computing earliest arrival flows, which Fleischer and Tardos~\cite{FleischerTardos1998}
analyze in the continuous time setting. These algorithms are based on shortest cycle
canceling. More precisely, an earliest arrival flow can be obtained by using a nonstandard
chain decomposition~$\Gamma$ of a static minimum cost flow~$x$ in~$\net{\sources}$
where~$\Gamma$ is the collection of chain flows used by the shortest cycle canceling
algorithm; see Algorithm~\ref{alg:EarliestArrival}.
\begin{algorithm}
\caption{Earliest arrival flow algorithm}
\label{alg:EarliestArrival}
\begin{algorithmic}[1]
\Require network~$\net{}$
\Ensure earliest arrival flow
\State $x^0\gets$ zero flow in~$\net{\sources}$
\State $\Gamma^0\gets\emptyset$
\State $i\gets0$
\While{there is a negative cycle in~$\net{\sources}_{x^i}$ with positive residual capacity}
	\State find shortest such cycle~$C^{i+1}$
	\State let~$\gamma^{i+1}$ be the static chain flow in~$\net{\sources}_{x^i}$ saturating~$C^{i+1}$
	\State $x^{i+1}\gets x^i+\gamma^{i+1}$
	\State $\Gamma^{i+1}\gets\Gamma^i\cup\{\gamma^{i+1}\}$
	\State $i\gets i+1$
\EndWhile
\State $\Gamma\gets\Gamma^i$
\State \Return~$[\Gamma]^\timehorizon$
\end{algorithmic}
\end{algorithm}
The resulting chain-decomposable flow~$[\Gamma]^\timehorizon$ is defined analogously to a
standard chain-decomposable (or temporally repeated) flow in
Definition~\ref{def:temporally-repeated}, where a standard chain decomposition~$\Gamma$ was
used.

\begin{Def}
\label{def:chain-decomposable-finite}
Let~$\Gamma$ be a nonstandard chain decomposition of some static circulation
in~$\net{\sources}$ such that, for each~$\gamma\in\Gamma$,~$C_\gamma$ contains
arc~$\supersource r_\gamma$ for some~$r_\gamma\in\sources$. Let~$P_\gamma$ denote the path
starting at~$r_\gamma$ that is obtained from~$C_\gamma$ by deleting~$\supersource$ and its
incident arcs. For each~$\gamma\in\Gamma$, the \defword{chain-decomposable
flow~$[\Gamma]^\timehorizon$ with time horizon~$\timehorizon$} sends flow at rate~$|\gamma|$
into path~$P_\gamma$ during the time interval~$\bigl[0,\timehorizon-\transit(P_\gamma)\bigr)$.
\end{Def}

This definition, however, raises certain
questions which we first illustrate for the example network in Figure~\ref{fig:counterexample}.

Given the network in Figure~\ref{fig:counterexample} with time horizon~$\timehorizon=4$,
Algorithm~\ref{alg:EarliestArrival} in its first iteration sends a chain flow~$\gamma_1$ of value~$|\gamma_1|=1$ along
path~$P_1\coloneqq s_1,v,w,t_2$ (or rather along
cycle~$C_{\gamma_1}\coloneqq\supersource,P_1,\supersource$). In its second iteration, it sends a chain flow~$\gamma_2$ of value~$|\gamma_2|=1$ along path~$P_2=s_2,w,v,t_1$ (i.e., along
cycle~$C_{\gamma_2}\coloneqq\supersource,P_2,\supersource$). Notice that~$P_2$ contains the backward arc~$wv$ with negative transit
time~$\transit_{wv}=-\transit_{vw}=-1$. The resulting chain-decomposable flow~$[\{\gamma_1,\gamma_2\}]^4$ is illustrated in Figures~\ref{fig:earliest-arrival-path} and~\ref{fig:earliest-arrival}.
\begin{figure}[p]
\centering
\begin{tikzpicture}[scale=0.8]
\begin{scope}[yshift=-0*3.5cm]
\node at (5,0.5) {$\timehorizon'=0$};
\node [node] (s1A) at (-3,1) {\footnotesize$s_1$};
\node [node] (vA) at (-1,1) {\footnotesize$v$};
\node [node] (wA) at (1,-1) {\footnotesize$w$};
\node [node] (t2A) at (3,-1) {\footnotesize$t_2$};
\draw [arc] (s1A) to node [above] {\footnotesize$1$} (vA);
\draw [arc] (vA) to node [right] {\footnotesize$1$} (wA);
\draw [arc] (wA) to node [above] {\footnotesize$1$} (t2A);
\fill [gray] (-3.3,1.9) rectangle +(-0.55,-0.55);
\fill [gray] (-3.3,1.3) rectangle +(-0.55,-0.55);
\fill [gray] (-3.3,0.7) rectangle +(-0.55,-0.55);
\end{scope}	
\begin{scope}[yshift=-1*3.5cm]
\node at (5,0.5) {$\timehorizon'=1$};
\node [node] (s1B) at (-3,1) {\footnotesize$s_1$};
\node [node] (vB) at (-1,1) {\footnotesize$v$};
\node [node] (wB) at (1,-1) {\footnotesize$w$};
\node [node] (t2B) at (3,-1) {\footnotesize$t_2$};
\draw [arc] (s1B) to node [above] {\footnotesize$1$} (vB);
\fill [gray] (s1B)+(0.2,-0.1) rectangle +(1.8,-0.3);
\draw [arc] (vB) to node [right] {\footnotesize$1$} (wB);
\draw [arc] (wB) to node [above] {\footnotesize$1$} (t2B);
\fill [gray] (-3.3,1.3) rectangle +(-0.55,-0.55);
\fill [gray] (-3.3,0.7) rectangle +(-0.55,-0.55);
\end{scope}	
\begin{scope}[yshift=-2*3.5cm]
\node at (5,0.5) {$\timehorizon'=2$};
\node [node] (s1C) at (-3,1) {\footnotesize$s_1$};
\node [node] (vC) at (-1,1) {\footnotesize$v$};
\node [node] (wC) at (1,-1) {\footnotesize$w$};
\node [node] (t2C) at (3,-1) {\footnotesize$t_2$};
\draw [arc] (s1C) to node [above] {\footnotesize$1$} (vC);
\fill [gray] (s1C)+(0.2,-0.1) rectangle +(1.8,-0.3);
\draw [arc] (vC) to node [right] {\footnotesize$1$} (wC);
\fill [gray,rotate=45] (vC)+(-0.1,-0.25) rectangle +(-0.25,-2.5);
\draw [arc] (wC) to node [above] {\footnotesize$1$} (t2C);
\fill [gray] (-3.3,0.7) rectangle +(-0.55,-0.55);
\end{scope}	
\begin{scope}[yshift=-3*3.5cm]
\node at (5,0.5) {$\timehorizon'=3$};
\node [node] (s1D) at (-3,1) {\footnotesize$s_1$};
\node [node] (vD) at (-1,1) {\footnotesize$v$};
\node [node] (wD) at (1,-1) {\footnotesize$w$};
\node [node] (t2D) at (3,-1) {\footnotesize$t_2$};
\draw [arc] (s1D) to node [above] {\footnotesize$1$} (vD);
\fill [gray] (s1D)+(0.2,-0.1) rectangle +(1.8,-0.3);
\draw [arc] (vD) to node [right] {\footnotesize$1$} (wD);
\fill [gray,rotate=45] (vD)+(-0.1,-0.25) rectangle +(-0.25,-2.5);
\draw [arc] (wD) to node [above] {\footnotesize$1$} (t2D);
\fill [gray] (wD)+(0.2,-0.1) rectangle +(1.8,-0.3);
\end{scope}	
\begin{scope}[yshift=-4*3.5cm]
\node at (5,0.5) {$\timehorizon'=4$};
\node [node] (s1E) at (-3,1) {\footnotesize$s_1$};
\node [node] (vE) at (-1,1) {\footnotesize$v$};
\node [node] (wE) at (1,-1) {\footnotesize$w$};
\node [node] (t2E) at (3,-1) {\footnotesize$t_2$};
\draw [arc] (s1E) to node [above] {\footnotesize$1$} (vE);
\draw [arc] (vE) to node [right] {\footnotesize$1$} (wE);
\fill [gray,rotate=45] (vE)+(-0.1,-0.25) rectangle +(-0.25,-2.5);
\draw [arc] (wE) to node [above] {\footnotesize$1$} (t2E);
\fill [gray] (wE)+(0.2,-0.1) rectangle +(1.8,-0.3);
\fill [gray] (3.3,-1.9) rectangle +(0.55,0.55);
\end{scope}	
\begin{scope}[yshift=-5*3.5cm]
\node at (5,0.5) {$\timehorizon'=5$};
\node [node] (s1F) at (-3,1) {\footnotesize$s_1$};
\node [node] (vF) at (-1,1) {\footnotesize$v$};
\node [node] (wF) at (1,-1) {\footnotesize$w$};
\node [node] (t2F) at (3,-1) {\footnotesize$t_2$};
\draw [arc] (s1F) to node [above] {\footnotesize$1$} (vF);
\draw [arc] (vF) to node [right] {\footnotesize$1$} (wF);
\draw [arc] (wF) to node [above] {\footnotesize$1$} (t2F);
\fill [gray] (wF)+(0.2,-0.1) rectangle +(1.8,-0.3);
\fill [gray] (3.3,-1.9) rectangle +(0.55,0.55);
\fill [gray] (3.3,-1.3) rectangle +(0.55,0.55);
\end{scope}	
\begin{scope}[yshift=-6*3.5cm]
\node at (5,0.5) {$\timehorizon'=6$};
\node [node] (s1G) at (-3,1) {\footnotesize$s_1$};
\node [node] (vG) at (-1,1) {\footnotesize$v$};
\node [node] (wG) at (1,-1) {\footnotesize$w$};
\node [node] (t2G) at (3,-1) {\footnotesize$t_2$};
\draw [arc] (s1G) to node [above] {\footnotesize$1$} (vG);
\draw [arc] (vG) to node [right] {\footnotesize$1$} (wG);
\draw [arc] (wG) to node [above] {\footnotesize$1$} (t2G);
\fill [gray] (3.3,-1.9) rectangle +(0.55,0.55);
\fill [gray] (3.3,-1.3) rectangle +(0.55,0.55);
\fill [gray] (3.3,-0.7) rectangle +(0.55,0.55);
\end{scope}	

\begin{scope}[yshift=-0*3.5cm,xshift=9.5cm]
\node [node] (s2a) at (-3,-1) {\footnotesize$s_2$};
\node [node] (va) at (-1,1) {\footnotesize$v$};
\node [node] (wa) at (1,-1) {\footnotesize$w$};
\node [node] (t1a) at (3,1) {\footnotesize$t_1$};
\draw [arc] (s2a) to node [above] {\footnotesize$3$} (wa);
\draw [arc,<-] (va) to node [right] {\footnotesize$-1$} (wa);
\draw [arc] (va) to node [above] {\footnotesize$3$} (t1a);
\fill [gray] (-3.3,-1.3) rectangle +(-0.55,0.55);
\end{scope}	
\begin{scope}[yshift=-1*3.5cm,xshift=9.5cm]
\node [node] (s2b) at (-3,-1) {\footnotesize$s_2$};
\node [node] (vb) at (-1,1) {\footnotesize$v$};
\node [node] (wb) at (1,-1) {\footnotesize$w$};
\node [node] (t1b) at (3,1) {\footnotesize$t_1$};
\draw [arc] (s2b) to node [above] {\footnotesize$3$} (wb);
\fill [gray] (s2b)+(0.2,-0.1) rectangle +(1.4,-0.4);
\draw [arc,<-] (vb) to node [right] {\footnotesize$-1$} (wb);
\draw [arc] (vb) to node [above] {\footnotesize$3$} (t1b);
\end{scope}	
\begin{scope}[yshift=-2*3.5cm,xshift=9.5cm]
\node [node] (s2c) at (-3,-1) {\footnotesize$s_2$};
\node [node] (vc) at (-1,1) {\footnotesize$v$};
\node [node] (wc) at (1,-1) {\footnotesize$w$};
\node [node] (t1c) at (3,1) {\footnotesize$t_1$};
\draw [arc] (s2c) to node [above] {\footnotesize$3$} (wc);
\fill [gray] (s2c)+(1.4,-0.1) rectangle +(2.6,-0.4);
\draw [arc,<-] (vc) to node [right] {\footnotesize$-1$} (wc);
\draw [arc] (vc) to node [above] {\footnotesize$3$} (t1c);
\end{scope}	
\begin{scope}[yshift=-3*3.5cm,xshift=9.5cm]
\node [node] (s2d) at (-3,-1) {\footnotesize$s_2$};
\node [node] (vd) at (-1,1) {\footnotesize$v$};
\node [node] (wd) at (1,-1) {\footnotesize$w$};
\node [node] (t1d) at (3,1) {\footnotesize$t_1$};
\draw [arc] (s2d) to node [above] {\footnotesize$3$} (wd);
\fill [gray] (s2d)+(2.6,-0.1) rectangle +(3.8,-0.4);
\draw [arc,<-] (vd) to node [right] {\footnotesize$-1$} (wd);
\fill [gray,rotate=45] (vd)+(-0.1,-0.25) rectangle +(-0.25,-2.5);
\draw [arc] (vd) to node [above] {\footnotesize$3$} (t1d);
\fill [gray] (vd)+(0.4,-0.1) rectangle +(1.6,-0.4);
\end{scope}	
\begin{scope}[yshift=-4*3.5cm,xshift=9.5cm]
\node [node] (s2e) at (-3,-1) {\footnotesize$s_2$};
\node [node] (ve) at (-1,1) {\footnotesize$v$};
\node [node] (we) at (1,-1) {\footnotesize$w$};
\node [node] (t1e) at (3,1) {\footnotesize$t_1$};
\draw [arc] (s2e) to node [above] {\footnotesize$3$} (we);
\draw [arc,<-] (ve) to node [right] {\footnotesize$-1$} (we);
\draw [arc] (ve) to node [above] {\footnotesize$3$} (t1e);
\fill [gray] (ve)+(1.5,-0.1) rectangle +(2.7,-0.4);
\end{scope}	
\begin{scope}[yshift=-5*3.5cm,xshift=9.5cm]
\node [node] (s2f) at (-3,-1) {\footnotesize$s_2$};
\node [node] (vf) at (-1,1) {\footnotesize$v$};
\node [node] (wf) at (1,-1) {\footnotesize$w$};
\node [node] (t1f) at (3,1) {\footnotesize$t_1$};
\draw [arc] (s2f) to node [above] {\footnotesize$3$} (wf);
\draw [arc,<-] (vf) to node [right] {\footnotesize$-1$} (wf);
\draw [arc] (vf) to node [above] {\footnotesize$3$} (t1f);
\fill [gray] (vf)+(2.6,-0.1) rectangle +(3.8,-0.4);
\end{scope}	
\begin{scope}[yshift=-6*3.5cm,xshift=9.5cm]
\node [node] (s2g) at (-3,-1) {\footnotesize$s_2$};
\node [node] (vg) at (-1,1) {\footnotesize$v$};
\node [node] (wg) at (1,-1) {\footnotesize$w$};
\node [node] (t1g) at (3,1) {\footnotesize$t_1$};
\draw [arc] (s2g) to node [above] {\footnotesize$3$} (wg);
\draw [arc,<-] (vg) to node [right] {\footnotesize$-1$} (wg);
\draw [arc] (vg) to node [above] {\footnotesize$3$} (t1g);
\fill [gray] (3.3,1.3) rectangle +(0.55,-0.55);
\end{scope}	
\end{tikzpicture}
\caption{Sending flow at rate~$1$ into paths~$P_1=s_1,v,w,t_2$ and~$P_2=s_2,w,v,t_1$ during time intervals~$[0,3)$ and~$[0,1)$, respectively}
\label{fig:earliest-arrival-path}
\end{figure}
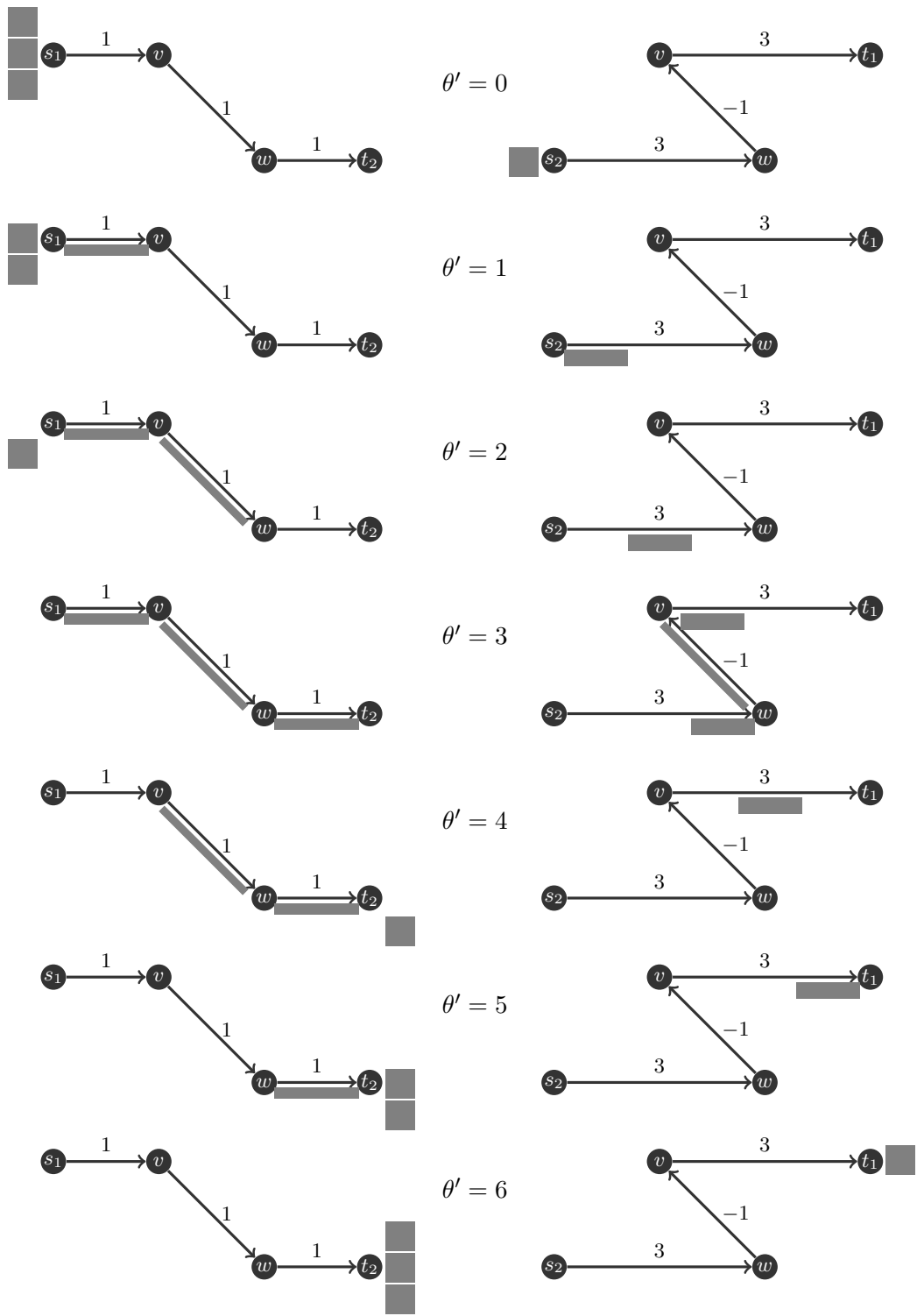
In Figure~\ref{fig:earliest-arrival-path} we illustrate the temporal development of flow along
the two paths~$P_1$ and~$P_2$ separately. The interesting snapshot is at
time~$\timehorizon'=3$ when three copies of the single flow unit traveling along path~$P_2$
can be seen. This is due to the negative transit time of arc~$wv$ which enables the flow unit
to travel simultaneously along arcs~$s_2w$ and~$vt_1$. But then again, the third copy of the
flow unit depicted on arc~$wv$ should be rather interpreted as a negative flow unit traveling
along the opposite arc~$vw$ such that, in total, there is still only one unit of flow.

Obviously, due to the discussed time travel issue, the flow over time along path~$P_2$
depicted in Figure~\ref{fig:earliest-arrival-path} is not a feasible flow over time on its
own. Together with the flow traveling along path~$P_1$, however, we obtain a feasible flow
over time since the negative flow unit on arc~$vw$ then cancels out a positive flow unit that
simultaneously travels along that arc as part of its path~$P_1$; see also
Figure~\ref{fig:earliest-arrival}.
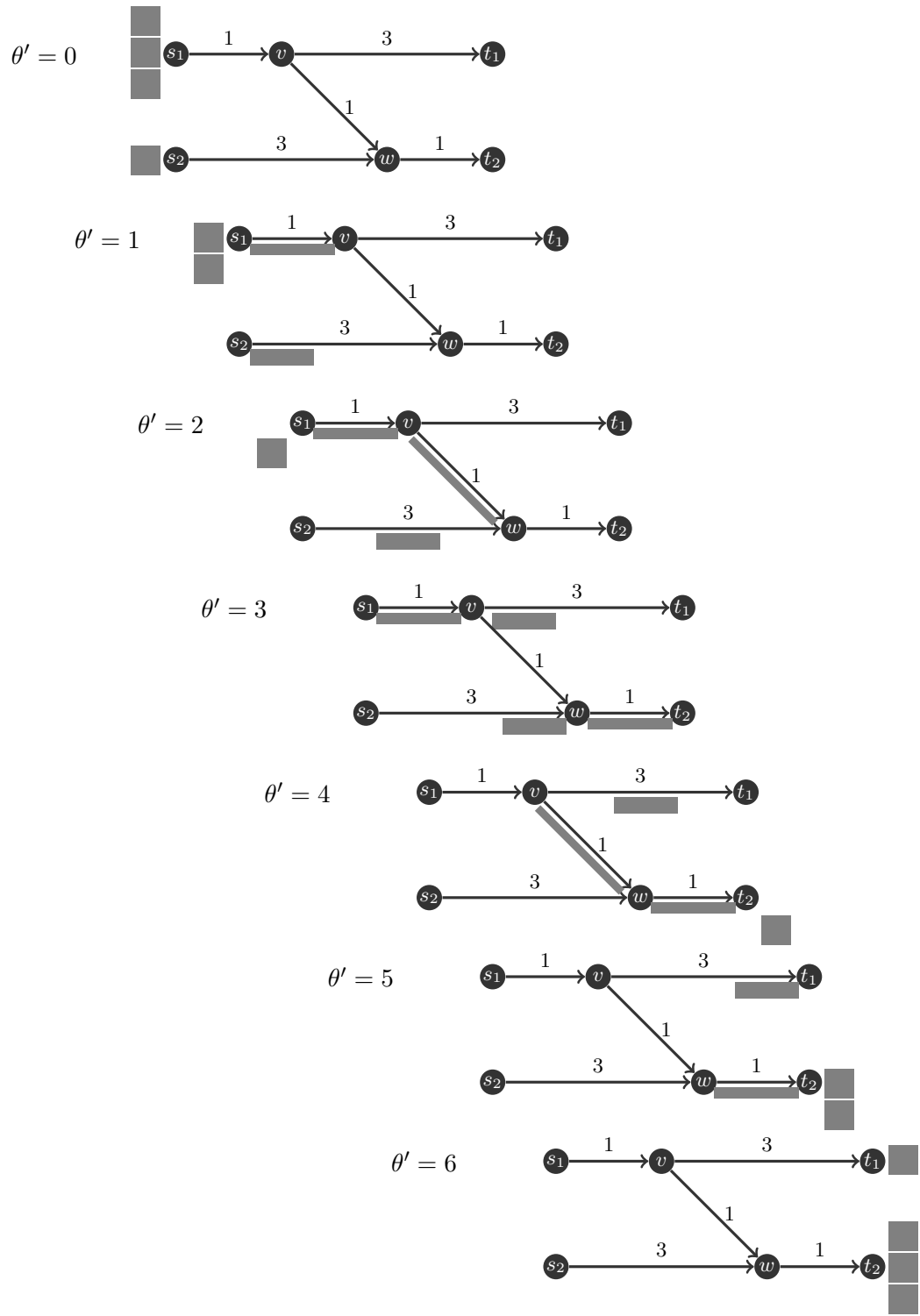
\begin{figure}[p]
\centering
\begin{tikzpicture}[scale=0.8]
\begin{scope}[yshift=-0*3.5cm,xshift=0*1.2cm]
\node at (-5.5,1) {$\timehorizon'=0$};
\node [node] (s1A) at (-3,1) {\footnotesize$s_1$};
\node [node] (vA) at (-1,1) {\footnotesize$v$};
\node [node] (wA) at (1,-1) {\footnotesize$w$};
\node [node] (t2A) at (3,-1) {\footnotesize$t_2$};
\node [node] (s2A) at (-3,-1) {\footnotesize$s_2$};
\node [node] (t1A) at (3,1) {\footnotesize$t_1$};
\draw [arc] (s2A) to node [above] {\footnotesize$3$} (wA);
\draw [arc] (vA) to node [above] {\footnotesize$3$} (t1A);
\draw [arc] (s1A) to node [above] {\footnotesize$1$} (vA);
\draw [arc] (vA) to node [right] {\footnotesize$1$} (wA);
\draw [arc] (wA) to node [above] {\footnotesize$1$} (t2A);
\fill [gray] (-3.3,1.9) rectangle +(-0.55,-0.55);
\fill [gray] (-3.3,1.3) rectangle +(-0.55,-0.55);
\fill [gray] (-3.3,0.7) rectangle +(-0.55,-0.55);
\fill [gray] (-3.3,-1.3) rectangle +(-0.55,0.55);
\end{scope}	
\begin{scope}[yshift=-1*3.5cm,xshift=1*1.2cm]
\node at (-5.5,1) {$\timehorizon'=1$};
\node [node] (s1B) at (-3,1) {\footnotesize$s_1$};
\node [node] (vB) at (-1,1) {\footnotesize$v$};
\node [node] (wB) at (1,-1) {\footnotesize$w$};
\node [node] (t2B) at (3,-1) {\footnotesize$t_2$};
\node [node] (s2B) at (-3,-1) {\footnotesize$s_2$};
\node [node] (t1B) at (3,1) {\footnotesize$t_1$};
\draw [arc] (s2B) to node [above] {\footnotesize$3$} (wB);
\fill [gray] (s2B)+(0.2,-0.1) rectangle +(1.4,-0.4);
\draw [arc] (vB) to node [above] {\footnotesize$3$} (t1B);
\draw [arc] (s1B) to node [above] {\footnotesize$1$} (vB);
\fill [gray] (s1B)+(0.2,-0.1) rectangle +(1.8,-0.3);
\draw [arc] (vB) to node [right] {\footnotesize$1$} (wB);
\draw [arc] (wB) to node [above] {\footnotesize$1$} (t2B);
\fill [gray] (-3.3,1.3) rectangle +(-0.55,-0.55);
\fill [gray] (-3.3,0.7) rectangle +(-0.55,-0.55);
\end{scope}	
\begin{scope}[yshift=-2*3.5cm,xshift=2*1.2cm]
\node at (-5.5,1) {$\timehorizon'=2$};
\node [node] (s1C) at (-3,1) {\footnotesize$s_1$};
\node [node] (vC) at (-1,1) {\footnotesize$v$};
\node [node] (wC) at (1,-1) {\footnotesize$w$};
\node [node] (t2C) at (3,-1) {\footnotesize$t_2$};
\node [node] (s2C) at (-3,-1) {\footnotesize$s_2$};
\node [node] (t1C) at (3,1) {\footnotesize$t_1$};
\draw [arc] (s2C) to node [above] {\footnotesize$3$} (wC);
\fill [gray] (s2C)+(1.4,-0.1) rectangle +(2.6,-0.4);
\draw [arc] (vC) to node [above] {\footnotesize$3$} (t1C);
\draw [arc] (s1C) to node [above] {\footnotesize$1$} (vC);
\fill [gray] (s1C)+(0.2,-0.1) rectangle +(1.8,-0.3);
\draw [arc] (vC) to node [right] {\footnotesize$1$} (wC);
\fill [gray,rotate=45] (vC)+(-0.1,-0.25) rectangle +(-0.25,-2.5);
\draw [arc] (wC) to node [above] {\footnotesize$1$} (t2C);
\fill [gray] (-3.3,0.7) rectangle +(-0.55,-0.55);
\end{scope}	
\begin{scope}[yshift=-3*3.5cm,xshift=3*1.2cm]
\node at (-5.5,1) {$\timehorizon'=3$};
\node [node] (s1D) at (-3,1) {\footnotesize$s_1$};
\node [node] (vD) at (-1,1) {\footnotesize$v$};
\node [node] (wD) at (1,-1) {\footnotesize$w$};
\node [node] (t2D) at (3,-1) {\footnotesize$t_2$};
\node [node] (s2D) at (-3,-1) {\footnotesize$s_2$};
\node [node] (t1D) at (3,1) {\footnotesize$t_1$};
\draw [arc] (s2D) to node [above] {\footnotesize$3$} (wD);
\fill [gray] (s2D)+(2.6,-0.1) rectangle +(3.8,-0.4);
\draw [arc] (vD) to node [above] {\footnotesize$3$} (t1D);
\fill [gray] (vD)+(0.4,-0.1) rectangle +(1.6,-0.4);
\draw [arc] (s1D) to node [above] {\footnotesize$1$} (vD);
\fill [gray] (s1D)+(0.2,-0.1) rectangle +(1.8,-0.3);
\draw [arc] (vD) to node [right] {\footnotesize$1$} (wD);
\draw [arc] (wD) to node [above] {\footnotesize$1$} (t2D);
\fill [gray] (wD)+(0.2,-0.1) rectangle +(1.8,-0.3);
\end{scope}	
\begin{scope}[yshift=-4*3.5cm,xshift=4*1.2cm]
\node at (-5.5,1) {$\timehorizon'=4$};
\node [node] (s1E) at (-3,1) {\footnotesize$s_1$};
\node [node] (vE) at (-1,1) {\footnotesize$v$};
\node [node] (wE) at (1,-1) {\footnotesize$w$};
\node [node] (t2E) at (3,-1) {\footnotesize$t_2$};
\node [node] (s2E) at (-3,-1) {\footnotesize$s_2$};
\node [node] (t1E) at (3,1) {\footnotesize$t_1$};
\draw [arc] (s2E) to node [above] {\footnotesize$3$} (wE);
\draw [arc] (vE) to node [above] {\footnotesize$3$} (t1E);
\fill [gray] (vE)+(1.5,-0.1) rectangle +(2.7,-0.4);
\draw [arc] (s1E) to node [above] {\footnotesize$1$} (vE);
\draw [arc] (vE) to node [right] {\footnotesize$1$} (wE);
\fill [gray,rotate=45] (vE)+(-0.1,-0.25) rectangle +(-0.25,-2.5);
\draw [arc] (wE) to node [above] {\footnotesize$1$} (t2E);
\fill [gray] (wE)+(0.2,-0.1) rectangle +(1.8,-0.3);
\fill [gray] (3.3,-1.9) rectangle +(0.55,0.55);
\end{scope}	
\begin{scope}[yshift=-5*3.5cm,xshift=5*1.2cm]
\node at (-5.5,1) {$\timehorizon'=5$};
\node [node] (s1F) at (-3,1) {\footnotesize$s_1$};
\node [node] (vF) at (-1,1) {\footnotesize$v$};
\node [node] (wF) at (1,-1) {\footnotesize$w$};
\node [node] (t2F) at (3,-1) {\footnotesize$t_2$};
\node [node] (s2F) at (-3,-1) {\footnotesize$s_2$};
\node [node] (t1F) at (3,1) {\footnotesize$t_1$};
\draw [arc] (s2F) to node [above] {\footnotesize$3$} (wF);
\draw [arc] (vF) to node [above] {\footnotesize$3$} (t1F);
\fill [gray] (vF)+(2.6,-0.1) rectangle +(3.8,-0.4);
\draw [arc] (s1F) to node [above] {\footnotesize$1$} (vF);
\draw [arc] (vF) to node [right] {\footnotesize$1$} (wF);
\draw [arc] (wF) to node [above] {\footnotesize$1$} (t2F);
\fill [gray] (wF)+(0.2,-0.1) rectangle +(1.8,-0.3);
\fill [gray] (3.3,-1.9) rectangle +(0.55,0.55);
\fill [gray] (3.3,-1.3) rectangle +(0.55,0.55);
\end{scope}	
\begin{scope}[yshift=-6*3.5cm,xshift=6*1.2cm]
\node at (-5.5,1) {$\timehorizon'=6$};
\node [node] (s1G) at (-3,1) {\footnotesize$s_1$};
\node [node] (vG) at (-1,1) {\footnotesize$v$};
\node [node] (wG) at (1,-1) {\footnotesize$w$};
\node [node] (t2G) at (3,-1) {\footnotesize$t_2$};
\node [node] (s2G) at (-3,-1) {\footnotesize$s_2$};
\node [node] (t1G) at (3,1) {\footnotesize$t_1$};
\draw [arc] (s2G) to node [above] {\footnotesize$3$} (wG);
\draw [arc] (vG) to node [above] {\footnotesize$3$} (t1G);
\draw [arc] (s1G) to node [above] {\footnotesize$1$} (vG);
\draw [arc] (vG) to node [right] {\footnotesize$1$} (wG);
\draw [arc] (wG) to node [above] {\footnotesize$1$} (t2G);
\fill [gray] (3.3,-1.9) rectangle +(0.55,0.55);
\fill [gray] (3.3,-1.3) rectangle +(0.55,0.55);
\fill [gray] (3.3,-0.7) rectangle +(0.55,0.55);
\fill [gray] (3.3,1.3) rectangle +(0.55,-0.55);
\end{scope}	
\end{tikzpicture}
\caption{Chain-decomposable earliest arrival flow with time horizon~$\timehorizon=6$, obtained by adding the flow along paths~$P_1=s_1,v,w,t_2$ and~$P_2=s_2,w,v,t_1$ depicted in Figure~\ref{fig:earliest-arrival-path}}
\label{fig:earliest-arrival}
\end{figure}

In general, a chain-decomposable flow can only be feasible if flow traveling along a backward
arc is always met by flow traveling along the corresponding forward arc at the same (or
higher) rate. In Figure~\ref{fig:infeasible} we give an example of a chain decomposition that does not meet this condition and therefore yields an infeasible chain-decomposable flow.
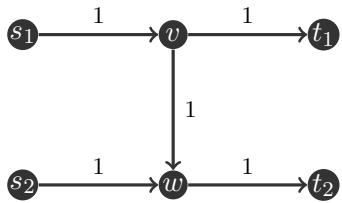
\begin{figure}[t]
\centering
\begin{tikzpicture}
\node [node] (s1) at (-2,1)	{$s_1$};
\node [node] (s2) at (-2,-1) {$s_2$};
\node [node] (v) at (0,1)	{$v$};
\node [node] (w) at (0,-1) {$w$};
\node [node] (t1) at (2,1)	{$t_1$};
\node [node] (t2) at (2,-1) {$t_2$};
\draw [arc] (s1) to node [above] {\footnotesize$1$} (v);
\draw [arc] (s2) to node [above] {\footnotesize$1$} (w);
\draw [arc] (v) to node [right] {\footnotesize$1$} (w);
\draw [arc] (v) to node [above] {\footnotesize$1$} (t1);
\draw [arc] (w) to node [above] {\footnotesize$1$} (t2);
\end{tikzpicture}
\caption{%
Unit capacity network with sources~$\sources=\{s_1,s_2\}$, sinks~$\sinks=\{t_1,t_2\}$, and
unit transit times; consider the static flow~$x$ with~$x_{vw}=0$ and~$x_a=1$ for all other
arcs~$a$; a nonstandard chain decomposition~$\Gamma$ may send one unit of flow along each of
the paths~$P_1=s_1,v,w,t_2$ and~$P_2=s_2,w,v,t_1$; for~$\timehorizon\geq3$, the corresponding
chain-decomposable flow~$[\Gamma]^\timehorizon$, however, is not feasible since~$P_2$ uses
backward arc~$wv$ before~$P_1$ can send flow along the corresponding forward arc~$vw$.
}
\label{fig:infeasible}
\end{figure}
The condition is, however, always fulfilled for the chain-decomposable flow found by
Algorithm~\ref{alg:EarliestArrival} as the shortest path distance from~$\supersource$ to any
node~$v$ of~$\net{\sources}_{x^i}$ is non-decreasing over the iterations~$i$. This follows
from the well-known result that augmenting flow along a shortest $s$-$t$-path in a residual
graph can neither decrease the shortest path distances from~$s$ nor the shortest path
distances to~$t$; see, e.g., Korte and Vygen~\cite[Chapter~9.4]{KorteVygen2018}.

\begin{theorem}[Wilkinson~\cite{Wilkinson1971}, Minieka~\cite{Minieka1973}]
The chain-decomposable flow~$[\Gamma]^\timehorizon$ returned by
Algorithm~\ref{alg:EarliestArrival} is an earliest arrival flow.
\end{theorem}

\begin{proof} 
We have already argued that~$[\Gamma]^\timehorizon$ is a feasible flow over
time. In order to show that the amount of flow arriving at the sinks before some point in
time~$\timehorizon'\leq\timehorizon$ is optimal, let~$i$ be the first iteration such
that~$\net{\sources}_{x^i}$ does not contain a cycle with positive residual capacity of length
at most~$\timehorizon'-\timehorizon$. Then~$x^i$ is a minimum cost circulation
in~$\net{\sources}$ for time horizon~$\timehorizon'$ with chain decomposition~$\Gamma^i$.
Moreover, the amount of flow arriving at the sinks in~$[\Gamma]^\timehorizon$ before
time~$\timehorizon'$ is equal to~$|[\Gamma^i]^{\timehorizon'}|$. As for the case of standard
chain decompositions in~\eqref{eq:max-flow-value}, we again obtain
\[
\bigl|[\Gamma^i]^{\timehorizon'}\bigr|
=\sum_{\gamma\in\Gamma^i}|\gamma|\bigl(\timehorizon'-\transit(P_\gamma)\bigr)
%=-\sum_{\gamma\in\Gamma^i}|\gamma|\transit(C_\gamma) 
=-\cost(x^i).
\]
Theorem~\ref{thm:max-flow-over-time} thus implies that the amount of flow arriving at the
sinks in~$[\Gamma]^\timehorizon$ before time~$\timehorizon'$ is maximal.
\end{proof}

Notice that the worst case running time of Algorithm~\ref{alg:EarliestArrival} is
not polynomial in the input size since shortest cycle canceling performed in its while loop
might require exponentially many iterations; see Zadeh~\cite{Zadeh1973}. Only recently, Disser
and Skutella~\cite{DisserSkutella2019} observed that an earliest arrival flow is indeed
NP-hard to find in some sense.

The flow over time found by Algorithm~\ref{alg:EarliestArrival} is not only an earliest
arrival flow but also a \defword{latest departure flow} with time horizon~$\timehorizon$. That
is, it maximizes the amount of flow leaving the sources~$\sources$ after time~$\timehorizon'$
for any~$\timehorizon'$ with~$0\leq\timehorizon'\leq\timehorizon$.

In the context of earliest arrival flows it is also interesting to consider flows over time
with infinite time horizon. Notice that, independently of how large the time
horizon~$\timehorizon$ is chosen, Algorithm~\ref{alg:EarliestArrival} terminates after finitely
many iterations as soon as there is no more source sink path with positive residual capacity
in~$\net{\sources}_{x^i}$. The final chain decomposition~$\Gamma$ then yields a flow over time
~$[\Gamma]$ with infinite time horizon maximizing the amount of flow that has arrived at the
sinks for any time~$\timehorizon\geq0$.

\begin{Def}
\label{def:chain-decomposable-infinite}
Let~$\Gamma$ be a nonstandard chain decomposition of some static circulation
in~$\net{\sources}$ such that, for each~$\gamma\in\Gamma$,~$C_\gamma$ contains
arc~$\supersource r_\gamma$ for some terminal~$r_\gamma\in\terminals$. Let~$P_\gamma$ denote
the path starting at~$r_\gamma$ that is obtained from~$C_\gamma$ by deleting~$\supersource$
and its incident arcs. For each~$\gamma\in\Gamma$, the \defword{chain-decomposable
flow~$[\Gamma]$ with infinite time horizon} sends flow at rate~$|\gamma|$ into path~$P_\gamma$
during the time interval~$[\transit_{\supersource r_\gamma},\infty)$.
\end{Def}

\subsubsection*{Further results on earliest arrival flows from the literature}

The following short literature review is an updated version of the corresponding review given
in~\cite{Skutella2009}. In view of the exponential worst case running time of
Algorithm~\ref{alg:EarliestArrival}, Hoppe and Tardos~\cite{HoppeTardos1994} present a fully
polynomial-time approximation scheme for the earliest arrival flow problem that is based on a
clever scaling trick. For a fixed $\varepsilon>0$, the computed flow over time has the
following property. For all~$\timehorizon\geq0$ simultaneously, the flow value at
time~$\timehorizon$ is at least $1-\varepsilon$ times the maximum possible value for~$\timehorizon$.

In a network with given supplies and demands at the terminals that may not be exceeded, flows
over time having the earliest arrival property do not necessarily exist; see
Fleischer~\cite{Fleischer2001}. For the case of several sources with given supplies and a
single sink, however, there is always such an earliest arrival flow. In this setting, Hajek
and Ogier~\cite{HajekOgier1984} give the first polynomial time algorithm for computing
earliest arrival flows in networks with zero transit times. Fleischer~\cite{Fleischer2001}
gives an algorithm with improved running time. For the case of multiple sinks with given
demands, Schmidt and Skutella~\cite{SchmidtSkutella2014} characterize which networks with
zero transit times always allow for earliest arrival flows. Groß, Kappmeier, Schmidt, and
Schmidt~\cite{GrossKappmeierSchmidtSchmidt2012} present tight approximation results for
earliest arrival flows in networks with arbitrary transit times and multiple sinks with given
demands. For such networks, Schlöter~\cite{Schloeter2019} proves that it is NP-hard to decide
whether an earliest arrival flow exists and gives an exact polynomial-space algorithm for a
special class of instances.

Baumann and Skutella~\cite{BaumannSkutella2009} give an algorithm that computes earliest
arrival flows for the case of several sources and arbitrary transit times and whose running
time is polynomially bounded in the input plus output size. Schlöter and
Skutella~\cite{SchloeterSkutella2017} present a refined variant of that algorithm that only
needs polynomial space. Fleischer and Skutella~\cite{FleischerSkutella2007} use so-called
condensed time-expanded networks to approximate such earliest arrival flows. They give a
fully polynomial-time approximation scheme that approximates the time delay as follows: For
every time~$\timehorizon\geq0$ simultaneously, the amount of flow that should have reached
the sink in an earliest arrival flow by time~$\timehorizon$, reaches the sink at latest at
time~$(1+\varepsilon)\timehorizon$. Tjandra~\cite{Tjandra2003} shows how to compute earliest
arrival transshipments in networks with time dependent supplies and capacities in time
polynomial in the time horizon and the total supply at sources.

Earliest arrival flows are motivated by applications related to evacuation. In the context of
emergency evacuation from buildings, Berlin~\cite{Berlin1979}. Jarvis and
Ratliff~\cite{JarvisRatliff1982} show that three different objectives of this optimization
problem can be achieved simultaneously: (i) Minimizing the total time needed to send the
supplies of all sources to the sink, (ii) fulfilling the earliest arrival property, and (iii)
minimizing the average time for all flow needed to reach the sink. Hamacher and
Tufecki~\cite{HamacherTufecki1987} study an evacuation problem and propose solutions which
further prevent unnecessary movement within a building. Dressler, Flötteröd, Lämmel, Nagel,
and Skutella~\cite{DresslerFloetteroedLaemmel+2011} combine network flow over time with
simulation techniques to find good evacuation strategies. Dressler et al.~\cite{ZET-ICEM2010}
use network flow over time techniques for finding good assignments of evacuees to emergency
exits.

%------------------------------------------------

\section{Lexicographically maximum flows over time}
\label{sec:lex-max-flows}

The efficient computation of lexicographically maximum flows over time is the main building
block of all known efficient transshipment over time algorithms.

\begin{Def}
Given an ordering of the~$k$ terminals~$r_k,r_{k-1},\dots,r_2,r_1$, a
\defword{lexicographically maximum flow over time} with time horizon~$\timehorizon$
lexicographically maximizes the flow amounts leaving the terminals in the given order.
\end{Def}

As the set of terminals~$\terminals$ contains both sources and sinks, notice that maximizing
the amount of flow leaving a terminal is equivalent to minimizing the amount of flow entering
the terminal. In other words, a lexicographically maximum flow over time lexicographically
minimizes the amount of flow entering the terminals in the given order.

Lexicographically maximum static flows have already been studied in the early 1970s by
Minieka~\cite{Minieka1973} and Megiddo~\cite{Megiddo1974}. They even proved the existence of a
feasible flow that simultaneously maximizes the total amount of flow leaving the~$i$ highest
priority terminals for each~$i=0,1,\dots,k-1$. Notice that such a flow is clearly
lexicographically maximum. Hoppe and Tardos~\cite{HoppeTardos2000} present an efficient
algorithm for computing a lexicographically maximum flow over time that also satisfies this
stronger property. 

We use the following notation: For a given ordering of terminals
$r_k,r_{k-1},\dots,r_2,r_1$ and~$i\in\{0,1,\dots,k-1\}$
let~$X_i\coloneqq\{r_k,r_{k-1},\dots,r_{i+1}\}$; moreover, let~$X_k\coloneqq\emptyset$.

\begin{theorem}[Hoppe, Tardos~\cite{HoppeTardos2000}]
\label{thm:lex-max-flow-over-time}
There is an efficient algorithm that, given an ordering of the~$k$
terminals~$r_k,r_{k-1},\dots,r_2,r_1$, finds a feasible flow over time with time
horizon~$\timehorizon$ such that the amount of flow sent from the sources in~$\sources\cap
X_i$ to the sinks in~$\sinks\setminus X_i$ is equal to~$\oo(X_i)$ and thus maximum, for
each~$i=0,1,\dots,k$.
\end{theorem}

Hoppe and Tardos' algorithm consists of a sequence of iterations~$i=1,\dots,k$, where static
minimum cost circulations~$x^i$ in networks~$\net{X_i}$ are computed. Point of departure is
the zero flow~$x^0$ in network~$\net{X_0}=\net{\terminals}$ where all arcs~$\supersource s$,~$s\in\sources$, are present but none of the sinks~$t\in\sinks$ has an arc~$t\supersource$.
In particular, all arcs in~$\net{X_0}$ have non-negative transit time such that the zero
flow~$x^0$ is a minimum cost circulation.

In each iteration~$i$, going from network~$\net{X_{i-1}}$ to~$\net{X_i}$, either
arc~$r_i\supersource$ is added if~$r_i$ is a sink, or arc~$\supersource r_i$ is removed
if~$r_i$ is a source. In both cases, the previous minimum cost circulation~$x^{i-1}$
in~$\net{X_{i-1}}$ is augmented into a minimum cost circulation~$x^i$ in~$\net{X_i}$ by
computing an appropriate minimum cost circulation~$y^i$ in the residual network of
flow~$x^{i-1}$, and setting~$x^i\coloneqq x^{i-1}+y^i$. Moreover, in each iteration the
algorithm computes a standard chain decomposition~$\Delta^i$ of~$y^i$. The flow over time
computed by the algorithm is then the chain-decomposable flow~$[\Gamma]$ where~$\Gamma$ is the
collection of all chain flows in~$\Delta^i$ for~$i=1,\dots,k$; see Algorithm~\ref{alg:lexmax}.
\begin{algorithm}[t]
\caption{Hoppe and Tardos' lex-max flow over time algorithm}
\label{alg:lexmax}
\begin{algorithmic}[1]
\Require ordering of terminals $r_k,r_{k-1},\dots,r_2,r_1$
\Ensure lexicographically maximum flow over time
\State $x^0 \gets$ zero flow in~$\net{X_0}$
\State $\Gamma^0 \gets \emptyset$
\For{$i \gets 1$ to~$k$}
	\If{$r_i\in\sinks$}
	    \State $y^i \gets$ minimum cost circulation in~$\net{X_i}_{x^{i-1}}$\label{step:sink}
	\ElsIf{$r_i\in\sources$}
	    \State $y^i \gets$ minimum cost circulation in~$\net{X_{i-1}}_{x^{i-1}}$ 
			with~$y^i_{r_i\supersource}=x^{i-1}_{\supersource r_i}$\label{step:source}
	\EndIf
	\State $x^i\gets x^{i-1}+y^i$
		\hfill\textcolor{gray}{// minimum cost circulation in~$\net{X_i}$}
	\State $\Delta^i \gets$ standard chain decomposition of~$y^i$
	\State $\Gamma^i \gets \Gamma^{i-1}\cup\Delta^i$
		\hfill\textcolor{gray}{// chain decomposition of~$x^i$}
\EndFor
\State $\Gamma\gets\Gamma^k$
\State \Return $[\Gamma]$
\end{algorithmic}
\end{algorithm}
We illustrate Algorithm~\ref{alg:lexmax} on a small example network in
Figures~\ref{fig:lexmax-example-static} and~\ref{fig:lexmax-example-dynamic}.

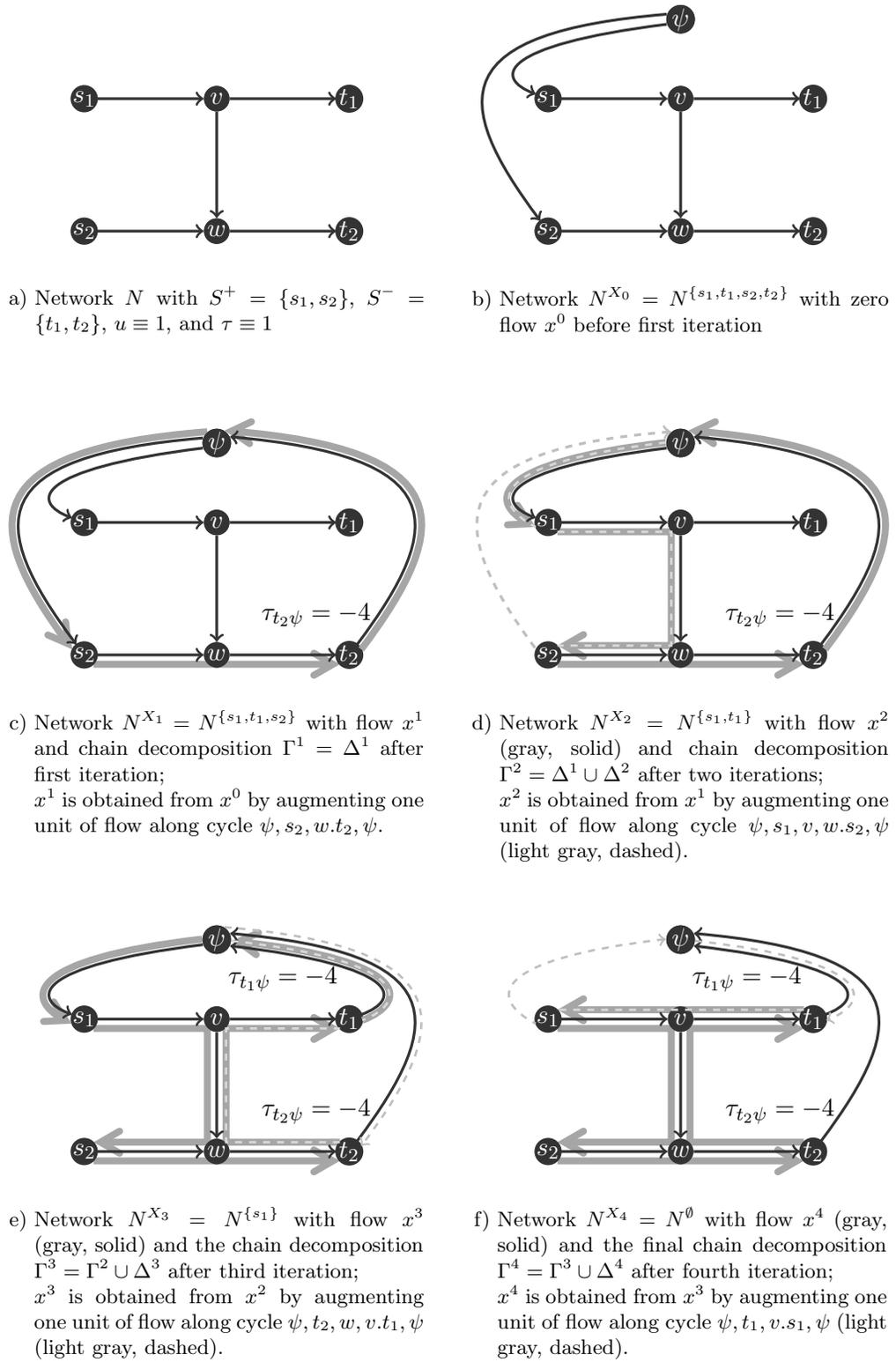
\begin{figure}[p]
\centering
\begin{tikzpicture}

\begin{scope}[xshift=-3.50cm,yshift=-0.1cm]
\node [node] (s1) at (-2,1)	{\footnotesize$s_1$};
\node [node] (s2) at (-2,-1) {\footnotesize$s_2$};
\node [node] (v) at (0,1)	{\footnotesize$v$};
\node [node] (w) at (0,-1) {\footnotesize$w$};
\node [node] (t1) at (2,1)	{\footnotesize$t_1$};
\node [node] (t2) at (2,-1) {\footnotesize$t_2$};

% \draw [->,gray!70,line width=3pt,>=angle 60] (s2.-45) -- (t2.225);
% \draw [gray!50!black] (s2)+(0.6,-0.45) node {\footnotesize$[0,\infty)$};
%
% \draw [->,gray!70,line width=3pt,>=angle 60] (t1.135) -- (s1.45);
% \draw [gray!50!black] (t1)+(-0.6,+0.45) node {\footnotesize$[\timehorizon,\infty)$};
%
% \draw [->,gray!70,line width=3pt,>=angle 60] (s1.-45) -- (v.225) -- (w.135) -- (s2.45);
% \draw [gray!50!black] (s1)+(0.6,-0.45) node {\footnotesize$[0,\infty)$};
%
% \draw [->,gray!70,line width=3pt,>=angle 60] (t2.135) -- (w.45) -- (v.-45) -- (t1.225);
% \draw [gray!50!black] (t2)+(-0.6,+0.45) node {\footnotesize$[\timehorizon,\infty)$};

\node [node] at (s1) {$s_1$};
\node [node] at (s2) {$s_2$};
\node [node] at (v) {$v$};
\node [node] at (w) {$w$};
\node [node] at (t1) {$t_1$};
\node [node] at (t2) {$t_2$};

\draw [arc] (s1) to (v);
\draw [arc] (s2) to (w);
\draw [arc] (v) to (w);
\draw [arc] (v) to (t1);
\draw [arc] (w) to (t2);

% \node [node] (psi) at (0,2.2) {$\supersource$};
% \draw [arc] (psi.200) .. controls (-2.8,1.8) and (-2.8,1.3) .. (s1);
% \draw [arc,<-] (psi.-25) .. controls (2.8,1.8) and (2.8,1.3) .. (t1);
% \draw [arc] (psi.160) .. controls (-3.5,2) and (-3.5,1) .. (s2);
% \draw [arc,<-] (psi.25) .. controls (3.5,2) and (3.5,1) .. (t2);

%\node [node] at (psi) {$\supersource$};

\node [anchor=north] at (0,-1.7) {
\footnotesize a)
\begin{minipage}[t]{0.42\textwidth}
Network~$\net{}$ with~$\sources=\{s_1,s_2\}$, $\sinks=\{t_1,t_2\}$, $u\equiv1$, and~$\transit\equiv1$
\end{minipage}
};
\end{scope}

\begin{scope}[xshift=3.50cm,yshift=-0.1cm]
\node [node] (s1) at (-2,1)	{\footnotesize$s_1$};
\node [node] (s2) at (-2,-1) {\footnotesize$s_2$};
\node [node] (v) at (0,1)	{\footnotesize$v$};
\node [node] (w) at (0,-1) {\footnotesize$w$};
\node [node] (t1) at (2,1)	{\footnotesize$t_1$};
\node [node] (t2) at (2,-1) {\footnotesize$t_2$};

% \draw [->,gray!70,line width=3pt,>=angle 60] (s2.-45) -- (t2.225);
% \draw [gray!50!black] (s2)+(0.6,-0.45) node {\footnotesize$[0,\infty)$};
%
% \draw [->,gray!70,line width=3pt,>=angle 60] (t1.135) -- (s1.45);
% \draw [gray!50!black] (t1)+(-0.6,+0.45) node {\footnotesize$[\timehorizon,\infty)$};
%
% \draw [->,gray!70,line width=3pt,>=angle 60] (s1.-45) -- (v.225) -- (w.135) -- (s2.45);
% \draw [gray!50!black] (s1)+(0.6,-0.45) node {\footnotesize$[0,\infty)$};
%
% \draw [->,gray!70,line width=3pt,>=angle 60] (t2.135) -- (w.45) -- (v.-45) -- (t1.225);
% \draw [gray!50!black] (t2)+(-0.6,+0.45) node {\footnotesize$[\timehorizon,\infty)$};

\node [node] at (s1) {$s_1$};
\node [node] at (s2) {$s_2$};
\node [node] at (v) {$v$};
\node [node] at (w) {$w$};
\node [node] at (t1) {$t_1$};
\node [node] at (t2) {$t_2$};

\draw [arc] (s1) to (v);
\draw [arc] (s2) to (w);
\draw [arc] (v) to (w);
\draw [arc] (v) to (t1);
\draw [arc] (w) to (t2);

\node [node] (psi) at (0,2.2) {$\supersource$};
\draw [arc] (psi.200) .. controls (-2.8,1.8) and (-2.8,1.3) .. (s1);
%\draw [arc,<-] (psi.-25) .. controls (2.8,1.8) and (2.8,1.3) .. (t1);
\draw [arc] (psi.160) .. controls (-3.5,2) and (-3.5,1) .. (s2);
%\draw [arc,<-] (psi.25) .. controls (3.5,2) and (3.5,1) .. (t2);

\node [node] at (psi) {$\supersource$};

\node [anchor=north] at (0,-1.7) {
\footnotesize b)
\begin{minipage}[t]{0.42\textwidth}
Network~$\net{X_0}=\net{\{s_1,t_1,s_2,t_2\}}$ with zero flow~$x^0$ before first iteration	
\end{minipage}
};
\end{scope}

\begin{scope}[xshift=-3.50cm,yshift=-1*6.5cm]
\node [node] (s1) at (-2,1)	{\footnotesize$s_1$};
\node [node] (s2) at (-2,-1) {\footnotesize$s_2$};
\node [node] (v) at (0,1)	{\footnotesize$v$};
\node [node] (w) at (0,-1) {\footnotesize$w$};
\node [node] (t1) at (2,1)	{\footnotesize$t_1$};
\node [node] (t2) at (2,-1) {\footnotesize$t_2$};

\draw [->,gray!70,line width=3pt,>=angle 60] (s2.-45) -- (t2.225);
%\draw [gray!50!black] (s2)+(0.6,-0.45) node {\footnotesize$[0,\infty)$};

% \draw [->,gray!70,line width=3pt,>=angle 60] (t1.135) -- (s1.45);
% \draw [gray!50!black] (t1)+(-0.6,+0.45) node {\footnotesize$[\timehorizon,\infty)$};
%
% \draw [->,gray!70,line width=3pt,>=angle 60] (s1.-45) -- (v.225) -- (w.135) -- (s2.45);
% \draw [gray!50!black] (s1)+(0.6,-0.45) node {\footnotesize$[0,\infty)$};
%
% \draw [->,gray!70,line width=3pt,>=angle 60] (t2.135) -- (w.45) -- (v.-45) -- (t1.225);
% \draw [gray!50!black] (t2)+(-0.6,+0.45) node {\footnotesize$[\timehorizon,\infty)$};

\node [node] at (s1) {$s_1$};
\node [node] at (s2) {$s_2$};
\node [node] at (v) {$v$};
\node [node] at (w) {$w$};
\node [node] at (t1) {$t_1$};
\node [node] at (t2) {$t_2$};

\draw [arc] (s1) to (v);
\draw [arc] (s2) to (w);
\draw [arc] (v) to (w);
\draw [arc] (v) to (t1);
\draw [arc] (w) to (t2);

\node [node] (psi) at (0,2.2) {$\supersource$};

\draw [->,gray!70,line width=3pt,>=angle 60] (psi.130) .. controls (-3.6,2.1) and (-3.6,1) .. (s2.155);
\draw [<-,gray!70,line width=3pt,>=angle 60] (psi.58) .. controls (3.6,2.1) and (3.63,1) .. (t2.25);

\draw [arc] (psi.200) .. controls (-2.8,1.8) and (-2.8,1.3) .. (s1);
%\draw [arc,<-] (psi.-25) .. controls (2.8,1.8) and (2.8,1.3) .. (t1);
\draw [arc] (psi.160) .. controls (-3.5,2) and (-3.5,1) .. (s2);
\draw [arc,<-] (psi.25) .. controls (3.5,2) and (3.5,1) .. (t2);
\node at (1.5,-0.4) {$\transit_{t_2\supersource}=-4$};

\node [node] at (psi) {$\supersource$};

\node [anchor=north] at (0,-1.7) {
\footnotesize c)
\begin{minipage}[t]{0.42\textwidth}
Network~$\net{X_1}=\net{\{s_1,t_1,s_2\}}$ with flow $x^1$ and chain decomposition $\Gamma^1=\Delta^1$ after first iteration; 

$x^1$~is obtained from~$x^0$ by augmenting one unit of flow along cycle~$\supersource,s_2,w.t_2,\supersource$. 	
\end{minipage}
};
\end{scope}

\begin{scope}[xshift=3.50cm,yshift=-1*6.5cm]
\node [node] (s1) at (-2,1)	{\footnotesize$s_1$};
\node [node] (s2) at (-2,-1) {\footnotesize$s_2$};
\node [node] (v) at (0,1)	{\footnotesize$v$};
\node [node] (w) at (0,-1) {\footnotesize$w$};
\node [node] (t1) at (2,1)	{\footnotesize$t_1$};
\node [node] (t2) at (2,-1) {\footnotesize$t_2$};

\draw [->,gray!70,line width=3pt,>=angle 60] (s2.-45) -- (t2.225);
%\draw [gray!50!black] (s2)+(0.6,-0.45) node {\footnotesize$[0,\infty)$};

% \draw [->,gray!70,line width=3pt,>=angle 60] (t1.135) -- (s1.45);
% \draw [gray!50!black] (t1)+(-0.6,+0.45) node {\footnotesize$[\timehorizon,\infty)$};
%
\draw [->,gray!70,line width=3pt,>=angle 60] (s1.-45) -- (v.225) -- (w.135) -- (s2.45);
%\draw [gray!50!black] (s1)+(0.6,-0.45) node {\footnotesize$[0,\infty)$};
%
% \draw [->,gray!70,line width=3pt,>=angle 60] (t2.135) -- (w.45) -- (v.-45) -- (t1.225);
% \draw [gray!50!black] (t2)+(-0.6,+0.45) node {\footnotesize$[\timehorizon,\infty)$};

\node [node] at (s1) {$s_1$};
\node [node] at (s2) {$s_2$};
\node [node] at (v) {$v$};
\node [node] at (w) {$w$};
\node [node] at (t1) {$t_1$};
\node [node] at (t2) {$t_2$};

\draw [arc] (s1) to (v);
\draw [arc] (s2) to (w);
\draw [arc] (v) to (w);
\draw [arc] (v) to (t1);
\draw [arc] (w) to (t2);

\node [node] (psi) at (0,2.2) {$\supersource$};

%\draw [->,gray!70,line width=3pt,>=angle 60] (psi.130) .. controls (-3.6,2.1) and (-3.6,1) .. (s2.155);
\draw [<-,gray!70,line width=3pt,>=angle 60] (psi.58) .. controls (3.6,2.1) and (3.63,1) .. (t2.25);

\draw [->,gray!70,line width=3pt,>=angle 60] (psi.178) .. controls (-2.9,1.9) and (-2.9,1.3) .. (s1.200);

\draw [arc] (psi.200) .. controls (-2.8,1.8) and (-2.8,1.3) .. (s1);
%\draw [arc,<-] (psi.-25) .. controls (2.8,1.8) and (2.8,1.3) .. (t1);
%\draw [arc] (psi.160) .. controls (-3.5,2) and (-3.5,1) .. (s2);
\draw [arc,<-] (psi.25) .. controls (3.5,2) and (3.5,1) .. (t2);

\draw [gray!30,line width=1pt,dashed] (psi.178) .. controls (-2.9,1.9) and (-2.9,1.3) .. (s1.200);
\draw [gray!30,line width=1pt,dashed] (s1.-45) -- (v.225) -- (w.135) -- (s2.45);
\draw [arc,gray!50,dashed,<-,] (psi.130) .. controls (-3.6,2.1) and (-3.6,1) .. (s2.155);

\node at (1.5,-0.4) {$\transit_{t_2\supersource}=-4$};

\node [node] at (psi) {$\supersource$};

\node [anchor=north] at (0,-1.7) {
\footnotesize d)
\begin{minipage}[t]{0.42\textwidth}
Network~$\net{X_2}=\net{\{s_1,t_1\}}$ with flow~$x^2$ (gray, solid) and chain decomposition $\Gamma^2=\Delta^1\cup\Delta^2$ after two iterations;

$x^2$~is obtained from~$x^1$ by augmenting one unit of flow along cycle~$\supersource,s_1,v,w.s_2,\supersource$ (light gray, dashed). 	
\end{minipage}
};
\end{scope}

\begin{scope}[xshift=-3.50cm,yshift=-2*7cm]
\node [node] (s1) at (-2,1)	{\footnotesize$s_1$};
\node [node] (s2) at (-2,-1) {\footnotesize$s_2$};
\node [node] (v) at (0,1)	{\footnotesize$v$};
\node [node] (w) at (0,-1) {\footnotesize$w$};
\node [node] (t1) at (2,1)	{\footnotesize$t_1$};
\node [node] (t2) at (2,-1) {\footnotesize$t_2$};

\draw [->,gray!70,line width=3pt,>=angle 60] (s2.-45) -- (t2.225);
%\draw [gray!50!black] (s2)+(0.6,-0.45) node {\footnotesize$[0,\infty)$};

% \draw [->,gray!70,line width=3pt,>=angle 60] (t1.135) -- (s1.45);
% \draw [gray!50!black] (t1)+(-0.6,+0.45) node {\footnotesize$[\timehorizon,\infty)$};
%
\draw [->,gray!70,line width=3pt,>=angle 60] (s1.-45) -- (v.225) -- (w.135) -- (s2.45);
%\draw [gray!50!black] (s1)+(0.6,-0.45) node {\footnotesize$[0,\infty)$};
%
\draw [->,gray!70,line width=3pt,>=angle 60] (t2.135) -- (w.45) -- (v.-45) -- (t1.225);
% \draw [gray!50!black] (t2)+(-0.6,+0.45) node {\footnotesize$[\timehorizon,\infty)$};

\node [node] at (s1) {$s_1$};
\node [node] at (s2) {$s_2$};
\node [node] at (v) {$v$};
\node [node] at (w) {$w$};
\node [node] at (t1) {$t_1$};
\node [node] at (t2) {$t_2$};

\draw [arc] (s1) to (v);
\draw [arc] (s2) to (w);
\draw [arc] (v) to (w);
\draw [arc] (v) to (t1);
\draw [arc] (w) to (t2);

\node [node] (psi) at (0,2.2) {$\supersource$};

% \draw [->,gray!70,line width=3pt,>=angle 60] (psi.130) .. controls (-3.6,2.1) and (-3.6,1) .. (s2.155);
% \draw [<-,gray!70,line width=3pt,>=angle 60] (psi.58) .. controls (3.6,2.1) and (3.63,1) .. (t2.25);

\draw [->,gray!70,line width=3pt,>=angle 60] (psi.178) .. controls (-2.9,1.9) and (-2.9,1.3) .. (s1.200);

\draw [<-,gray!70,line width=3pt,>=angle 60] (psi.-5) .. controls (2.92,1.89) and (2.92,1.26) .. (t1.-5);

\draw [arc] (psi.200) .. controls (-2.8,1.8) and (-2.8,1.3) .. (s1);
\draw [arc,<-] (psi.-25) .. controls (2.8,1.8) and (2.8,1.3) .. (t1);
%\draw [arc] (psi.160) .. controls (-3.5,2) and (-3.5,1) .. (s2);
\draw [arc,<-] (psi.25) .. controls (3.5,2) and (3.5,1) .. (t2);
\node at (1.5,-0.4) {$\transit_{t_2\supersource}=-4$};
\node at (1,1.6) {$\transit_{t_1\supersource}=-4$};

\draw [->,gray!50,line width=1pt,dashed] (psi.58) .. controls (3.6,2.1) and (3.63,1) .. (t2.25);
\draw [gray!30,line width=1pt,dashed] (t2.135) -- (w.45) -- (v.-45) -- (t1.225);
\draw [gray!30,line width=1pt,dashed] (psi.-5) .. controls (2.92,1.89) and (2.92,1.26) .. (t1.-5);

\node [node] at (psi) {$\supersource$};

\node [anchor=north] at (0,-1.7) {
\footnotesize e)
\begin{minipage}[t]{0.42\textwidth}
Network~$\net{X_3}=\net{\{s_1\}}$ with flow~$x^3$ (gray, solid) and the chain decomposition\linebreak $\Gamma^3=\Gamma^2\cup\Delta^3$ after third iteration;

$x^3$~is obtained from~$x^2$ by augmenting one unit of flow along cycle~$\supersource,t_2,w,v.t_1,\supersource$ (light gray, dashed). 	
\end{minipage}
};
\end{scope}

\begin{scope}[xshift=3.50cm,yshift=-2*7cm]
\node [node] (s1) at (-2,1)	{\footnotesize$s_1$};
\node [node] (s2) at (-2,-1) {\footnotesize$s_2$};
\node [node] (v) at (0,1)	{\footnotesize$v$};
\node [node] (w) at (0,-1) {\footnotesize$w$};
\node [node] (t1) at (2,1)	{\footnotesize$t_1$};
\node [node] (t2) at (2,-1) {\footnotesize$t_2$};

\draw [->,gray!70,line width=3pt,>=angle 60] (s2.-45) -- (t2.225);
%\draw [gray!50!black] (s2)+(0.6,-0.45) node {\footnotesize$[0,\infty)$};

\draw [->,gray!70,line width=3pt,>=angle 60] (t1.135) -- (s1.45);
% \draw [gray!50!black] (t1)+(-0.6,+0.45) node {\footnotesize$[\timehorizon,\infty)$};
%
\draw [->,gray!70,line width=3pt,>=angle 60] (s1.-45) -- (v.225) -- (w.135) -- (s2.45);
%\draw [gray!50!black] (s1)+(0.6,-0.45) node {\footnotesize$[0,\infty)$};
%
\draw [->,gray!70,line width=3pt,>=angle 60] (t2.135) -- (w.45) -- (v.-45) -- (t1.225);
% \draw [gray!50!black] (t2)+(-0.6,+0.45) node {\footnotesize$[\timehorizon,\infty)$};

\node [node] at (s1) {$s_1$};
\node [node] at (s2) {$s_2$};
\node [node] at (v) {$v$};
\node [node] at (w) {$w$};
\node [node] at (t1) {$t_1$};
\node [node] at (t2) {$t_2$};

\draw [arc] (s1) to (v);
\draw [arc] (s2) to (w);
\draw [arc] (v) to (w);
\draw [arc] (v) to (t1);
\draw [arc] (w) to (t2);

\node [node] (psi) at (0,2.2) {$\supersource$};

% \draw [->,gray!70,line width=3pt,>=angle 60] (psi.130) .. controls (-3.6,2.1) and (-3.6,1) .. (s2.155);
% \draw [<-,gray!70,line width=3pt,>=angle 60] (psi.58) .. controls (3.6,2.1) and (3.63,1) .. (t2.25);

% \draw [->,gray!70,line width=3pt,>=angle 60] (psi.178) .. controls (-2.9,1.9) and (-2.9,1.3) .. (s1.200);

% \draw [<-,gray!70,line width=3pt,>=angle 60] (psi.-5) .. controls (2.92,1.89) and (2.92,1.26) .. (t1.-5);

% \draw [arc] (psi.200) .. controls (-2.8,1.8) and (-2.8,1.3) .. (s1);
\draw [arc,<-] (psi.-25) .. controls (2.8,1.8) and (2.8,1.3) .. (t1);
% \draw [arc] (psi.160) .. controls (-3.5,2) and (-3.5,1) .. (s2);
\draw [arc,<-] (psi.25) .. controls (3.5,2) and (3.5,1) .. (t2);
\node at (1.5,-0.4) {$\transit_{t_2\supersource}=-4$};
\node at (1,1.6) {$\transit_{t_1\supersource}=-4$};

\draw [->,gray!50,line width=1pt,dashed] (psi.-5) .. controls (2.92,1.89) and (2.92,1.26) .. (t1.-5);
\draw [gray!30,line width=1pt,dashed] (t1.135) -- (s1.45);
\draw [<-,gray!50,line width=1pt,dashed] (psi.178) .. controls (-2.9,1.9) and (-2.9,1.3) .. (s1.200);

\node [node] at (psi) {$\supersource$};

\node [anchor=north] at (0,-1.7) {
\footnotesize f)
\begin{minipage}[t]{0.42\textwidth}
Network~$\net{X_4}=\net{\emptyset}$ with flow~$x^4$ (gray, solid) and the final chain decomposition $\Gamma^4=\Gamma^3\cup\Delta^4$ after fourth iteration;

$x^4$~is obtained from~$x^3$ by augmenting one unit of flow along cycle~$\supersource,t_1,v.s_1,\supersource$ (light gray, dashed). 	
\end{minipage}
}; 
\end{scope}
\end{tikzpicture}
\caption{%
Static flows~$x^i$ and chain decompositions~$\Gamma^i$ computed by Algorithm~\ref{alg:lexmax} for input network~$\net{}$ (see~Fig.~\ref{fig:lexmax-example-static}a) with
unit capacities, unit transit times, time horizon~$\timehorizon=4$, and terminal order~$s_1,t_1,s_2,t_2$
}
\label{fig:lexmax-example-static}
\end{figure}
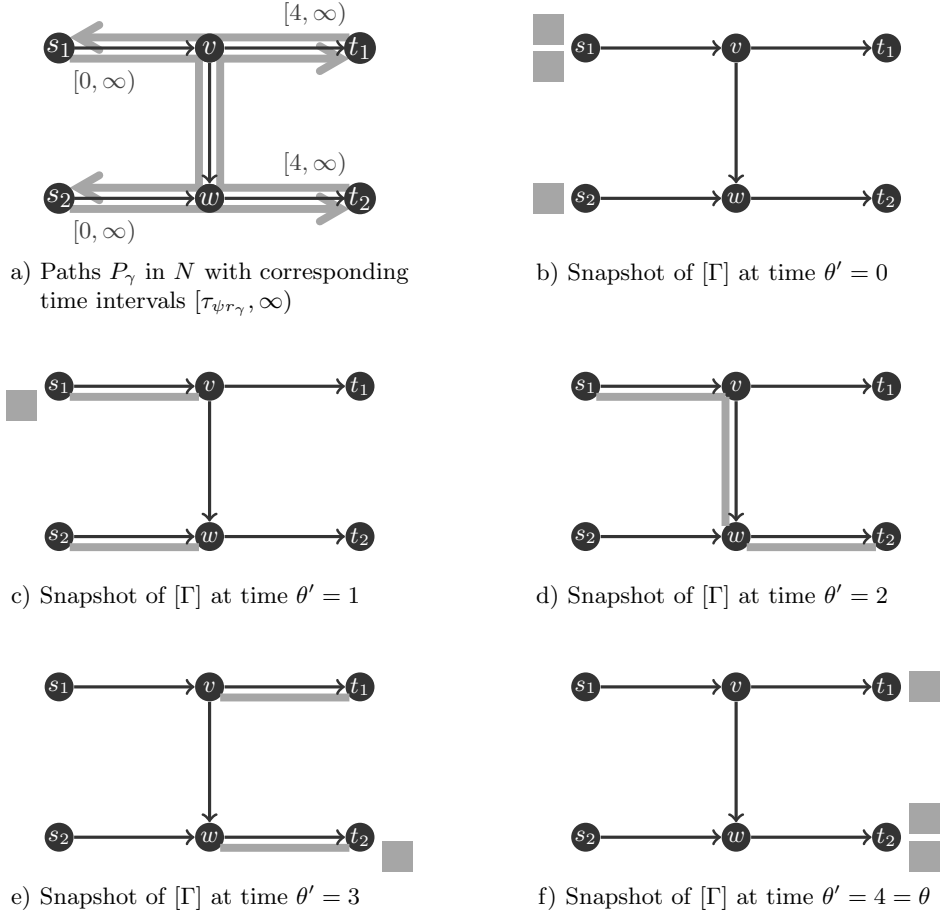
\begin{figure}[t]
\centering
\begin{tikzpicture}

\begin{scope}[xshift=-3.5cm,yshift=-0*7cm]
\node [node] (s1) at (-2,1)	{\footnotesize$s_1$};
\node [node] (s2) at (-2,-1) {\footnotesize$s_2$};
\node [node] (v) at (0,1)	{\footnotesize$v$};
\node [node] (w) at (0,-1) {\footnotesize$w$};
\node [node] (t1) at (2,1)	{\footnotesize$t_1$};
\node [node] (t2) at (2,-1) {\footnotesize$t_2$};

\draw [->,gray!70,line width=3pt,>=angle 60] (s2.-45) -- (t2.225);
\draw [gray!50!black] (s2)+(0.6,-0.45) node {\footnotesize$[0,\infty)$};

\draw [->,gray!70,line width=3pt,>=angle 60] (t1.135) -- (s1.45);
\draw [gray!50!black] (t1)+(-0.6,+0.45) node {\footnotesize$[4,\infty)$};

\draw [->,gray!70,line width=3pt,>=angle 60] (s1.-45) -- (v.225) -- (w.135) -- (s2.45);
\draw [gray!50!black] (s1)+(0.6,-0.45) node {\footnotesize$[0,\infty)$};

\draw [->,gray!70,line width=3pt,>=angle 60] (t2.135) -- (w.45) -- (v.-45) -- (t1.225);
\draw [gray!50!black] (t2)+(-0.6,+0.45) node {\footnotesize$[4,\infty)$};

\node [node] at (s1) {$s_1$};
\node [node] at (s2) {$s_2$};
\node [node] at (v) {$v$};
\node [node] at (w) {$w$};
\node [node] at (t1) {$t_1$};
\node [node] at (t2) {$t_2$};

\draw [arc] (s1) to (v);
\draw [arc] (s2) to (w);
\draw [arc] (v) to (w);
\draw [arc] (v) to (t1);
\draw [arc] (w) to (t2);

\node [anchor=north] at (0,-1.7) {
\footnotesize a)
\begin{minipage}[t]{0.35\textwidth}
Paths $P_\gamma$ in~$\net{}$ with corresponding time intervals~$[\transit_{\supersource r_\gamma},\infty)$
\end{minipage}
};
\end{scope}

\begin{scope}[xshift=3.5cm,yshift=-0*7cm]
\node [node] (s1) at (-2,1)	{\footnotesize$s_1$};
\node [node] (s2) at (-2,-1) {\footnotesize$s_2$};
\node [node] (v) at (0,1)	{\footnotesize$v$};
\node [node] (w) at (0,-1) {\footnotesize$w$};
\node [node] (t1) at (2,1)	{\footnotesize$t_1$};
\node [node] (t2) at (2,-1) {\footnotesize$t_2$};

\draw [arc] (s1) to (v);
\draw [arc] (s2) to (w);
\draw [arc] (v) to (w);
\draw [arc] (v) to (t1);
\draw [arc] (w) to (t2);

\fill [gray!70] (-2.3,1.05) rectangle +(-0.4,0.4);
\fill [gray!70] (-2.3,0.95) rectangle +(-0.4,-0.4);

\fill [gray!70] (-2.3,-1.2) rectangle +(-0.4,0.4);

\node [anchor=north] at (0,-1.7) {
\footnotesize b)
\begin{minipage}[t]{0.35\textwidth}
Snapshot of~$[\Gamma]$ at time~$\timehorizon'=0$
\end{minipage}
};
\end{scope}

\begin{scope}[xshift=-3.5cm,yshift=-4.5cm]
\node [node] (s1) at (-2,1)	{\footnotesize$s_1$};
\node [node] (s2) at (-2,-1) {\footnotesize$s_2$};
\node [node] (v) at (0,1)	{\footnotesize$v$};
\node [node] (w) at (0,-1) {\footnotesize$w$};
\node [node] (t1) at (2,1)	{\footnotesize$t_1$};
\node [node] (t2) at (2,-1) {\footnotesize$t_2$};

\draw [gray!70,line width=3pt,>=angle 60] (s2.-45) -- (w.225);

\draw [gray!70,line width=3pt,>=angle 60] (s1.-45) -- (v.225);

\draw [arc] (s1) to (v);
\draw [arc] (s2) to (w);
\draw [arc] (v) to (w);
\draw [arc] (v) to (t1);
\draw [arc] (w) to (t2);

\fill [gray!70] (-2.3,0.95) rectangle +(-0.4,-0.4);

\node [anchor=north] at (0,-1.5) {
\footnotesize c)
\begin{minipage}[t]{0.35\textwidth}
Snapshot of~$[\Gamma]$ at time~$\timehorizon'=1$
\end{minipage}
};
\end{scope}

\begin{scope}[xshift=3.5cm,yshift=-4.5cm]
\node [node] (s1) at (-2,1)	{\footnotesize$s_1$};
\node [node] (s2) at (-2,-1) {\footnotesize$s_2$};
\node [node] (v) at (0,1)	{\footnotesize$v$};
\node [node] (w) at (0,-1) {\footnotesize$w$};
\node [node] (t1) at (2,1)	{\footnotesize$t_1$};
\node [node] (t2) at (2,-1) {\footnotesize$t_2$};

\draw [gray!70,line width=3pt,>=angle 60] (w.-45) -- (t2.225);

\draw [gray!70,line width=3pt,>=angle 60] (s1.-45) -- (v.225);

\draw [gray!70,line width=3pt,>=angle 60] (v.225) -- (w.135);

\draw [arc] (s1) to (v);
\draw [arc] (s2) to (w);
\draw [arc] (v) to (w);
\draw [arc] (v) to (t1);
\draw [arc] (w) to (t2);

\node [anchor=north] at (0,-1.5) {
\footnotesize d)
\begin{minipage}[t]{0.35\textwidth}
Snapshot of~$[\Gamma]$ at time~$\timehorizon'=2$
\end{minipage}
};
\end{scope}

\begin{scope}[xshift=-3.5cm,yshift=-8.5cm]
\node [node] (s1) at (-2,1)	{\footnotesize$s_1$};
\node [node] (s2) at (-2,-1) {\footnotesize$s_2$};
\node [node] (v) at (0,1)	{\footnotesize$v$};
\node [node] (w) at (0,-1) {\footnotesize$w$};
\node [node] (t1) at (2,1)	{\footnotesize$t_1$};
\node [node] (t2) at (2,-1) {\footnotesize$t_2$};

\draw [gray!70,line width=3pt,>=angle 60] (w.-45) -- (t2.225);

\draw [gray!70,line width=3pt,>=angle 60] (v.-45) -- (t1.225);

\draw [arc] (s1) to (v);
\draw [arc] (s2) to (w);
\draw [arc] (v) to (w);
\draw [arc] (v) to (t1);
\draw [arc] (w) to (t2);

\fill [gray!70] (2.3,-1.05) rectangle +(0.4,-0.4);

\node [anchor=north] at (0,-1.5) {
\footnotesize e)
\begin{minipage}[t]{0.35\textwidth}
Snapshot of~$[\Gamma]$ at time~$\timehorizon'=3$
\end{minipage}
};
\end{scope}

\begin{scope}[xshift=3.5cm,yshift=-8.5cm]
\node [node] (s1) at (-2,1)	{\footnotesize$s_1$};
\node [node] (s2) at (-2,-1) {\footnotesize$s_2$};
\node [node] (v) at (0,1)	{\footnotesize$v$};
\node [node] (w) at (0,-1) {\footnotesize$w$};
\node [node] (t1) at (2,1)	{\footnotesize$t_1$};
\node [node] (t2) at (2,-1) {\footnotesize$t_2$};

\draw [arc] (s1) to (v);
\draw [arc] (s2) to (w);
\draw [arc] (v) to (w);
\draw [arc] (v) to (t1);
\draw [arc] (w) to (t2);

\fill [gray!70] (2.3,-1.05) rectangle +(0.4,-0.4);
\fill [gray!70] (2.3,-0.95) rectangle +(0.4,0.4);

\fill [gray!70] (2.3,0.8) rectangle +(0.4,0.4);

\node [anchor=north] at (0,-1.5) {
\footnotesize f)
\begin{minipage}[t]{0.35\textwidth}
Snapshot of~$[\Gamma]$ at time~$\timehorizon'=4=\timehorizon$
\end{minipage}
};
\end{scope}
\end{tikzpicture}
\caption{%
Chain-decomposable flow~$[\Gamma]$ computed by Algorithm~\ref{alg:lexmax} for input network~$\net{}$
and time horizon~$\timehorizon=4$ (cp.~Figure~\ref{fig:lexmax-example-static})
}
\label{fig:lexmax-example-dynamic}
\end{figure}

Our analysis of Algorithm~\ref{alg:lexmax} resembles the original analysis given by Hoppe and
Tardos~\cite{HoppeTardos2000}; in particular we use the same or similar notation. There are,
however, some crucial differences. While Hoppe and Tardos use cuts in time-expanded networks
to prove optimality of the chain-decomposable flow~$[\Gamma]$, we provide a shorter and more
direct proof based on Corollary~\ref{cor:max-flow-over-time}. As a consequence, our analysis
works for arbitrary (rational or irrational) time horizon and transit times. Also Fleischer
and Tardos~\cite{FleischerTardos1998} observe that the algorithm works for the case of an
irrational time horizon, however still assuming integral (or rational) arc transit times.

Before we start to analyze the chain-decomposable flow~$[\Gamma]$, we first prove that the
circulations~$x^i$ in~$\net{X_i}$ have indeed minimum cost and~$\Gamma^i$ is a chain
decomposition of~$x^i$; notice that this is also stated in Algorithm~\ref{alg:lexmax} as
comments.

\begin{lemma}
\label{lem:min-cost-circulation}%\label{lem:property}
For~$i=0,\dots,k$, $x^i$ is a feasible flow and a
minimum cost circulation in~$\net{X_i}$ with chain decomposition~$\Gamma^i$;
moreover, all chain flows~$\gamma\in\Delta^i$ use arc~$r_i\supersource$.
\end{lemma}

\begin{proof}
We use induction on~$i$. For~$i=0$, the statement is true since all arcs
of~$\net{X_0}=\net{\terminals}$ have non-negative transit times, and~$x^0$ is the zero flow
with empty chain decomposition~$\Gamma^0$. In iteration~$i\geq1$, we may assume
that~$x^{i-1}$ is feasible and a minimum cost circulation in~$\net{X_{i-1}}$. In particular, no
cycle in the residual graph~$\net{X_{i-1}}_{x^{i-1}}$ has negative cost.

If~$r_i$ is a sink,
arc~$r_i\supersource$ is added in iteration~$i$. Then, the minimum cost circulation~$y^i$
in~$\net{X_i}_{x^{i-1}}$ cancels all negative cycles such that~$x^i=x^{i-1}+y^i$ is indeed a
minimum cost circulation; moreover, all these negative cycles must contain arc~$r_i\supersource$. 

If~$r_i$ is a source,~$y^i$ cancels all flow on arc~$\supersource r_i$ such
that~$x^i=x^{i-1}+y^i$ is indeed a feasible flow in~$\net{X_i}$ where arc~$\supersource r_i$
has been removed. Moreover, since~$y^i$ is chosen to have minimum cost, no negative cycles
can occur in~$\net{X_i}_{x^i}$, and~$x^i$ is thus a minimum cost circulation. Moreover, since
no cycle in the residual graph~$\net{X_{i-1}}_{x^{i-1}}$ has negative cost,~$y^i$ can only
send flow along cycles containing arc~$r_i\supersource$.
\end{proof}

Notice that the final minimum cost circulation~$x^k$ in~$\net{X_k}=\net{\sinks}$ is the zero
flow, since~$\net{\sinks}$ does not contain cycles of non-positive cost.
Therefore~$\Gamma=\Gamma^k$ is a chain decomposition of the zero flow. This explains why the
resulting chain-decomposable flow~$[\Gamma]$ has finite time horizon even though it never
stops sending flow into its source sink paths (see
Definition~\ref{def:chain-decomposable-infinite}). As we show below, all flow in the network
cancels out after time~$\timehorizon$ such that~$[\Gamma]$ actually has time
horizon~$\timehorizon$. 

Moreover, in view of the discussion in
Section~\ref{sec:earliest-arrival-flows} on canceling flow along backward arcs with negative
transit time, we also need to argue that~$[\Gamma]$ is indeed a feasible flow over time. To
this end, for~$i=0,\dots,k$ and~$v\in V$, let node label~$p^i(v)$ denote the minimum cost of
a~$\supersource$-$v$-path in~$\net{X_i}_{x^i}$.

\begin{lemma}
\label{lem:feasibility}
The following holds for all~ $i=1,\dots,k$:
\begin{enumerate}[(i)]
\item\label{lem:labels} $p^i(v)\geq p^{i-1}(v)\geq0$ for all~$v\in V$.
\item\label{lem:path-lengths} For all~$\gamma\in\Delta^i$, if~$\gamma$ contains node~$v\in
V$, then the length of the $\supersource$-$v$-subpath of~$\gamma$ lies in the
interval~$\bigl[p^{i-1}(v),\min\{p^i(v),\timehorizon\}\bigr]$.
\end{enumerate}
\end{lemma}

In order to prove Lemma~\ref{lem:feasibility}, we will consider the residual
network~$\net{X_i}_{x^{i-1}}$, but first need to clarify a subtle issue: If~$r_i\in\sinks$,
then network~$\net{X_{i}}$ is obtained from~$\net{X_{i-1}}$ by adding arc~$r_i\supersource$.
In particular,~$x^{i-1}$ can be naturally interpreted as a circulation in~$\net{X_{i}}$ by
simply setting the flow on arc~$r_i\supersource$ to zero. This is indeed what is meant in
Step~\ref{step:sink} of Algorithm~\ref{alg:lexmax} when the corresponding residual
network~$\net{X_i}_{x^{i-1}}$ is being considered.

If~$r_i\in\sources$, however, network~$\net{X_{i}}$ is obtained from~$\net{X_{i-1}}$ by
deleting arc~$\supersource r_i$. Since circulation~$x^{i-1}$ may send a positive amount of
flow through arc~$\supersource r_i$, it cannot simply be interpreted as a feasible flow
in~$\net{X_{i}}$. Nevertheless, by slightly overloading notation, we
let~$\net{X_i}_{x^{i-1}}$ denote the residual network of the flow in~$\net{X_i}$ that is
obtained from~$x^{i-1}$ by ignoring the flow on arc~$\supersource r_i$, that is, setting it
to zero. Notice that the resulting flow in~$\net{X_i}$ is no longer a circulation in general
but rather an~$r_i$-$\supersource$-flow of value~$x^{i-1}_{\supersource r_i}$. An alternative
way of formulating Step~\ref{step:source} of Algorithm~\ref{alg:lexmax} is thus to find a minimum
cost~$\supersource$-$r_i$-flow of value~$x^{i-1}_{\supersource r_i}$ in~$\net{X_i}_{x^{i-1}}$
by which the above mentioned $r_i$-$\supersource$-flow is augmented to yield the minimum cost
circulation~$x^i$ in~$\net{X_i}$.

\begin{proof}[of Lemma~\ref{lem:feasibility}]
Let~$q^i(v)$ denote the minimum cost of a~$\supersource$-$v$-path
in~$\net{X_{i}}_{x^{i-1}}$.
Remember that~$\net{X_{i}}$ is obtained from~$\net{X_{i-1}}$ by either deleting
arc~$\supersource r_i$ (if~$r_i\in\sources$) or adding arc~$r_i\supersource$
(if~$r_i\in\sinks$). Neither of these changes can create a shorter~$\supersource$-$v$-path
for any node~$v$ in the residual network~$\net{X_{i}}_{x^{i-1}}$; thus~$q^i(v)\geq
p^{i-1}(v)$ for all~$v\in V$.

As discussed above, circulation~$x^i$ can be obtained from~$x^{i-1}$ by adding a minimum cost
flow in the residual network~$\net{X_{i}}_{x^{i-1}}$. This minimum cost flow can be obtained
by repeatedly augmenting flow in the residual network from~$\supersource$ along a
shortest~$\supersource$-$r_i$-path (if~$r_i\in\sources$) or cycle (if~$r_i\in\sinks$), respectively. It is
a well-known result in network flow theory that such flow augmentations cannot decrease the
shortest path distances from~$\supersource$ in the residual network; see, for
example,~\cite[Chapter~9.4]{KorteVygen2018}. Thus,~$p^i(v)\geq q^i(v)\geq p^{i-1}(v)$.

Moreover,~$p^0(v)\geq0$ since network~$\net{X_0}_{x^0}=\net{\terminals}$ only contains arcs of non-negative cost. This implies the non-negativity of the node labels stated in~\eqref{lem:labels} and therefore concludes the proof of this part of the lemma.

In order to prove~\eqref{lem:path-lengths}, we again use the node labels~$q^i(v)$. Since
the~$\supersource$-$v$-subpath of~$\gamma$ lives in the residual
network~$\net{X_{i}}_{x^{i-1}}$, its length is at least~$q^i(v)\geq p^{i-1}(v)$. On the other
hand, a flow-carrying $\supersource$-$v$-path of length strictly larger than~$p^i(v)$ would
induce a negative cycle in~$\net{X_i}_{x^i}$, contradicting
Lemma~\ref{lem:min-cost-circulation}. 

It remains to show that the length of the~$\supersource$-$v$-subpath of~$\gamma$ is bounded
from above by~$\timehorizon$. We distinguish two cases.

\emph{First case:} $x^i$ sends a positive amount of flow through node~$v$. Since~$x^i$ is a
minimum cost circulation in~$\net{X_i}$, all flow-carrying cycles in~$\net{X_i}$ must have
non-positive cost. The only negative cost arcs in~$\net{X_i}$ are those
entering~$\supersource$ with transit time~$-\timehorizon$. Therefore, there must exist a
flow-carrying $v$-$\supersource$-path in~$\net{X_i}$ whose last arc~$r_j\supersource$ has
transit time~$\transit_{r_j\supersource}=-\timehorizon$. The corresponding backward
$\supersource$-$v$-path lives in the residual network~$\net{X_i}_{x^i}$ and has cost at
most~$\timehorizon$ since~$\supersource r_j$ is its only arc with positive transit time. As a
consequence, the length of a shortest~$\supersource$-$v$-path in~$\net{X_i}_{x^i}$
is~$p^i(v)\leq\timehorizon$. And, as we argued above,~$p^i(v)$ is an upper bound on the
length of the~$\supersource$-$v$-subpath of~$\gamma$.

\emph{Second case:} $x^i$ does not send any flow through node~$v$. Notice
that~$x^i=x^{i-1}+y^i$ and~$y^i$ does send positive flow through~$v$ (e.g., along~$\gamma$).
Thus, this flow must cancel out flow that is sent through~$v$ by~$x^{i-1}$. Since~$x^{i-1}$
is a minimum cost circulation in~$\net{X_{i-1}}$, we can use the same arguments as in the
first case to show that positive flow through~$v$ always implies the existence of a backward
path in the residual graph of length at most~$\timehorizon$. Since~$y^i$ is a minimum cost
flow in the residual network~$\net{X_{i}}_{x^{i-1}}$, it (iteratively) augments flow along
shortest paths or cycles. In particular, the length of the~$\supersource$-$v$-subpath
of~$\gamma$ is bounded from above by~$\timehorizon$.
\end{proof}

Lemma~\ref{lem:feasibility}\,\eqref{lem:path-lengths}, in particular, states that flow being
sent by the chain decomposable flow~$[\Gamma]$ along the source sink path given by
chain~$\gamma\in\Delta^i\subseteq\Gamma$ arrives at node~$v$ from some time~$\timehorizon'$
on, where~$\timehorizon'\in\bigl[p^{i-1}(v),\min\{p^i(v),\timehorizon\}\bigr]$.
With the help of Lemma~\ref{lem:feasibility}, we can now prove
feasibility of the chain-decomposable flow~$[\Gamma]$.

\begin{theorem}
\label{thm:lex-max-feasibility}
The chain-decomposable flow~$[\Gamma]$ is a feasible flow over time in~$\net{}$ with time
horizon~$\timehorizon$.
\end{theorem}

\begin{proof}
First of all,~$[\Gamma]$ satisfies flow conservation at all non-terminal nodes, since all
paths~$P_\gamma$,~$\gamma\in\Gamma$, start and end at a terminal, and all flow arriving at an
intermediate node~$v$ while traveling along some path~$P_\gamma$ immediately continues its
journey on the next arc leaving~$v$.

Next we consider an arc~$a=vw$ in~$\net{}$ and argue that the flow rate entering arc~$vw$ is
always within the interval~$[0,u_{vw}]$ and equal to~$0$ before time~$0$ and from
time~$\timehorizon$ on.

First of all, by Lemma~\ref{lem:feasibility}, the chain-decomposable
flow~$[\Gamma]$ cannot send any flow into arc~$vw$ before time~$0$. Moreover, after
time~$\min\{p^k(v),\timehorizon\}$ all flow on arc~$vw$ cancels out since from that time on
the flow rates on all source sink paths~$P_{\gamma}$ given by chains~$\gamma\in\Gamma$ that contain
arc~$vw$ sum up to~$x^k_{vw}=0$.

For some point in time~$\timehorizon'\in\bigl[0,\min\{p^k(v),\timehorizon\}\bigr]$,
consider~$i$ with~$\timehorizon'\in[p^i(v),p^{i+1}(v)]$. Then, again by
Lemma~\ref{lem:feasibility}, the flow rate entering arc~$vw$ at
time~$\timehorizon'$ is within the
interval~$\bigl[\min\{x^i_{vw},x^{i+1}_{vw}\},\max\{x^i_{vw},x^{i+1}_{vw}\}\bigr]\subseteq[0,u_{vw}]$.
\end{proof}

We finally analyze the amount of flow that~$[\Gamma]$ sends out of a source or into a sink. We
use the following notation: For a terminal~$r\in\terminals$, let~$|[\Gamma]|_r$ denote the net
amount of flow leaving~$r$ in~$[\Gamma]$. Moreover, let~$\Gamma_{\supersource r}$
and~$\Gamma_{r\supersource}$ denote the subsets of chain flows in~$\Gamma$ that use
arc~$\supersource r$ and arc~$r\supersource$, respectively.

\begin{lemma}%[Hoppe, Tardos~\cite{HoppeTardos2000}]
\label{lem:flow-at-terminal}
$|[\Gamma]|_{r_i}=\cost(\Delta^i)$ for all~$i=1,\dots,k$.
\end{lemma}

\begin{proof}
Notice that the only chain flows in~$\Gamma$ that contribute to the net amount of flow leaving
terminal~$r_i$ in~$[\Gamma]$ are those in~$\Gamma_{\supersource r_i}\cup\Gamma_{r_i\supersource}$.
Since~$\Gamma$ is a chain decomposition of the zero flow, the chain flows on
arcs~$\supersource r_i$ and~$r_i\supersource$ cancel out, that is,
\begin{align}
\sum_{\gamma\in\Gamma_{\supersource r_i}}|\gamma|
=\sum_{\gamma\in\Gamma_{r_i\supersource}}|\gamma|.
\label{eq:canceling_out}
\end{align}
Notice that~$\Gamma_{r_i\supersource}=\Delta^i$ by Lemma~\ref{lem:min-cost-circulation},
since every chain flow~$\gamma\in\Gamma$ contains exactly one arc entering~$\supersource$. We
now distinguish two cases.

\emph{First case:} $r_i\in\sources$.
For every chain flow~$\gamma\in\Gamma_{\supersource r_i}$, the chain-decomposable
flow~$[\Gamma]$ starts to send flow at rate~$|\gamma|$ from source~$r_i$ into the network at
time~$0$. Moreover, for every chain flow~$\gamma\in\Gamma_{r_i\supersource}$, flow at
rate~$|\gamma|$ arrives at source~$r_i$ from time~$\transit(\gamma)$ on. Thus, due
to equation~\eqref{eq:canceling_out}, flow rates leaving and entering~$r_i$ cancel out from
time~$\bar{\timehorizon}\coloneqq\max\{0,\max_{\gamma\in\Gamma_{r_i\supersource}}\transit(\gamma)\}$
on. As a consequence, the net amount of flow leaving~$r_i$ in~$[\Gamma]$ is equal to
\begin{align*}
|[\Gamma]|_{r_i}
=\bar{\timehorizon}\sum_{\gamma\in\Gamma_{\supersource r_i}}|\gamma|-\sum_{\gamma\in\Gamma_{r_i\supersource}}\bigl(\bar{\timehorizon}-\transit(\gamma)\bigr)|\gamma|
\stackrel{\eqref{eq:canceling_out}}{=}\sum_{\gamma\in\Gamma_{r_i\supersource}}\transit(\gamma)|\gamma|
=\cost(\Gamma_{r_i\supersource})
=\cost(\Delta^i).
\end{align*}

\emph{Second case:} $r_i\in\sinks$.
For every chain flow~$\gamma\in\Gamma_{r_i\supersource}$, flow at rate~$|\gamma|$ arrives at
sink~$r_i$ from time~$\timehorizon+\transit(\gamma)$ on. Moreover, for every chain
flow~$\gamma\in\Gamma_{\supersource r_i}$, flow leaves~$r_i$ at rate~$|\gamma|$ from
time~$\timehorizon=\transit_{\supersource r_i}$ on. Thus, due to
equation~\eqref{eq:canceling_out}, flow rates entering and leaving~$r_i$ cancel out from
time~$\bar{\timehorizon}\coloneqq\timehorizon+\max\{0,\max_{\gamma\in\Gamma_{r_i\supersource}}
\transit(\gamma)\}$ on. As a consequence, the net amount of flow leaving~$r_i$ in~$[\Gamma]$
is equal to
\begin{align*}
|[\Gamma]|_{r_i}
=(\bar{\timehorizon}-\timehorizon)\sum_{\gamma\in\Gamma_{\supersource r_i}}|\gamma|-\sum_{\gamma\in\Gamma_{r_i\supersource}}\bigl(\bar{\timehorizon}-\timehorizon-\transit(\gamma)\bigr)|\gamma|
\stackrel{\eqref{eq:canceling_out}}{=}\sum_{\gamma\in\Gamma_{r_i\supersource}}\transit(\gamma)|\gamma|
=\cost(\Delta^i).
\end{align*}
This concludes the proof.
\end{proof}

Based on Lemma~\ref{lem:flow-at-terminal}, it is now easy to prove that~$[\Gamma]$ sends the
required amount of flow from~$\sources\cap X_i$ to~$\sinks\setminus X_i$ (cp.\
Theorem~\ref{thm:lex-max-flow-over-time}).

\begin{cor}
\label{cor:lex-max}
$\sum_{r\in X_i}|[\Gamma]|_r=\oo(X_i)$ for all~$i=0,\dots,k$.
\end{cor}

\begin{proof}
By flow conservation and Lemma~\ref{lem:flow-at-terminal},
\begin{align*}
\sum_{r\in X_i}|[\Gamma]|_r
=\sum_{j=i+1}^k|[\Gamma]|_{r_j}
=-\sum_{j=1}^i|[\Gamma]|_{r_j}
=-\sum_{j=1}^i\cost(\Delta^j).
\end{align*}
Notice that~$\bigcup_{j=1}^i\Delta^j=\Gamma^i$, which is a chain decomposition of the minimum cost circulation~$x^i$ in~$\net{X_i}$. As a consequence
\begin{align*}
\sum_{r\in X_i}|[\Gamma]|_r
=-\cost(x^i)
=\oo(X_i).
\end{align*}
This concludes the proof.
\end{proof}

We conclude this section with the proof of the main Theorem~\ref{thm:lex-max-flow-over-time} and a remark.

\begin{proof}[of Theorem~\ref{thm:lex-max-flow-over-time}]
By Theorem~\ref{thm:lex-max-feasibility}, Algorithm~\ref{alg:lexmax} computes a feasible flow
over time in network~$\net{}$ with time horizon~$\timehorizon$. Corollary~\ref{cor:lex-max}
states that the computed chain-decomposable flow~$[\Gamma]$ is a lexicographically maximum
flow over time. Finally notice that the running time of Algorithm~\ref{alg:lexmax} is dominated by the~$k$
minimum cost circulation computations, one in each iteration of its for-loop, which can be
done in strongly polynomial time; see, e.g.,~\cite{KorteVygen2018}.
\end{proof}

\begin{rem}
Notice that the chain-decomposable flow~$[\Gamma]$ computed by Algorithm~\ref{alg:lexmax} has
the following favorable properties: For integral arc capacities, the values of all chain
flows in~$\Gamma$ are integral such that flow rates on arcs are always integral
in~$[\Gamma]$. Moreover, for the case of integral arc transit times, flow rates on arcs only
change at integral points in time in~$[\Gamma]$.
\end{rem}

%------------------------------------------------

\section{Transshipments over time and submodular function minimization}
\label{sec:transshipments}

In this section we turn to the transshipment over time problem that is defined as follows.

\begin{Def}
Given a network~$\net{}$, time horizon~$\timehorizon$, supplies~$b_r\geq0$, $r\in\sources$,
and demands~$b_r\leq0$, $r\in\sinks$, with~$\sum_{r\in\terminals}b_r=0$, a feasible solution
to the \defword{transshipment over time problem} is a feasible flow over time with time
horizon~$\timehorizon$ that sends~$b_r$ units of flow out of each source~$r\in\sources$
and~$-b_r$ units into each sink~$r\in\sinks$.
\end{Def}

For the static transshipment problem, already Ford and Fulkerson observed that there is a
feasible flow satisfying all supplies and demands if and only if the maximum amount of flow
that can be sent from sources in~$\sources\cap X$ to sinks in~$\sinks\setminus X$ is at
least~$b(X)\coloneqq\sum_{r\in X}b_r$. This characterization is a straightforward consequence
of the max-flow min-cut theorem in the extended network with a super-source~$s$, connected to
all sources~$r\in\sources$ via an arc~$sr$ of capacity~$u_{sr}=b_r$, and a supersink~$t$
reachable from all sinks~$r\in\sinks$ via an arc~$rt$ of capacity~$u_{rt}=-b_r$.

Unfortunately, the resulting reduction of the static transshipment problem to a static
maximum~$s$-$t$-flow problem cannot be generalized to transshipments over time. Here the
problem is that the total amount of flow sent from a super-source through one of the
sources~$r\in\sources$ into the network cannot be easily bounded by~$b_r$ since the capacity
of the arc connecting the super-source to~$r$ only bounds the flow rate but not the amount of
flow sent through this arc.

Nevertheless, Klinz~\cite{Klinz1994} observed that the feasibility criterion mentioned above
also holds for the transshipment over time problem.

\begin{theorem}[Klinz~\cite{Klinz1994}]
\label{thm:klinz}
A feasible solution to the transshipment over time problem exists if and only if
\begin{align}
b(X)\leq\oo(X)\qquad\text{for each~$X\subseteq\terminals$.}
\label{eq:klinz}
\end{align}
\end{theorem}

While the necessity of Condition~\eqref{eq:klinz} is straightforward, the sufficiency is less
obvious. In the discrete time model considered by Klinz~\cite{Klinz1994}, it is merely a
consequence of the sufficiency of the cut condition in the time-expanded network. We give a
short proof that does not rely on time-expansion and therefore proves the result for networks
with arbitrary time horizon and arc transit times.

\begin{proof}
The two main ingredients of the proof are the submodularity of the
function~$\oo:2^{\terminals}\to\R$ (Theorem~\ref{thm:submodular}) and the existence of
lexicographically maximum flows over time that simultaneously maximize the total amount of
flow leaving the~$i$ highest priority terminals for each~$i=0,1,\dots,k-1$
(Theorem~\ref{thm:lex-max-flow-over-time}). Notice that we gave proofs of both results for
networks with arbitrary time horizon and arc transit times.

Consider the \defword{base polytope}~$B(\oo)$ of submodular function~$\oo:2^{\terminals}\to\R$ given by
\[
B(\oo)\coloneqq\bigl\{z\in\R^{\terminals}\mid z(X)\leq\oo(X)~\forall\,X\subset\terminals,~z(\terminals)=\oo(\terminals)\bigr\}
\]
Edmonds~\cite{Edmonds1970} and Shapley~\cite{Shapley1971} observed that
every vertex~$\bar{z}$ of the base polytope corresponds to an
ordering~$r_k,r_{k-1},\dots,r_1$ of the terminals from which it can be obtained as follows via
a greedy procedure:
\begin{align*}
\bar{z}_i\coloneqq\oo\bigl(\{r_k,r_{k-1},\dots,r_i\}\bigr)-\oo\bigl(\{r_k,r_{k-1},\dots,r_{i+1}\}\bigr)\qquad\text{for~$i=k,\dots,1$.}
\end{align*}
Notice that, by Theorem~\ref{thm:lex-max-flow-over-time}, vertex~$\bar{z}$ is exactly the vector of supplies and demands satisfied by a lexicographically maximum flow over time. 

Moreover, by definition of~$B(\oo)$, a supply and demand vector~$b\in\R^{\terminals}$
with
\[
\sum_{r\in\terminals}b_r=0=\oo(\terminals)
\]
satisfies Condition~\eqref{eq:klinz} if and only if~$b\in B(\oo)$. Such a vector~$b$ is
therefore a convex combination of vertices of~$B(\oo)$. And as such,~$b$ is the supply and
demand vector satisfied by the corresponding convex combination of lexicographically maximum
flows over time. Since a convex combination of feasible flows over time is again feasible, we
have thus proved that there exists a feasible solution to the transshipment over time problem,
under the assumption that the given supplies and demands satisfy Condition~\eqref{eq:klinz}.
\end{proof}

The proof of Theorem~\ref{thm:klinz} raises the question whether and how a feasible solution
to the transshipment over time problem can be found, if one exists. First of all, Condition~\eqref{eq:klinz} can be tested efficiently via submodular function minimization.

\begin{theorem}[Hoppe, Tardos~\cite{HoppeTardos2000}]
\label{thm:sfm}
There is a strongly polynomial time algorithm that determines whether a transshipment over time problem has a feasible solution.
\end{theorem}

\begin{proof}
In view of the proof of Theorem~\ref{thm:klinz}, testing feasibility boils down to testing
membership of the supply and demand vector~$b$ in the base polytope~$B(\oo)$. It is well known
that the separation problem over the base polytope is a submodular function minimization
problem; see, e.g., McCormick~\cite{McCormick2005}. More precisely,~$b\in B(\oo)$ if and only
if the minimum value of the function
\begin{align}	
X\mapsto\oo(X)-b(X)\qquad\text{for~$X\subseteq\terminals$,}
\label{eq:submodular-minus-modular}
\end{align}
is equal to zero (the minimum value is at most zero since~$\oo(\emptyset)-b(\emptyset)=0$).

Notice that function~\eqref{eq:submodular-minus-modular} is submodular, because it is the sum
of the submodular function~$X\mapsto\oo(X)$ and the modular function~$X\mapsto-b(X)$. Finally,
submodular function minimization can be solved in strongly polynomial time. 
\end{proof}

The first strongly polynomial algorithm for submodular function minimization is due to
Grötschel, Lovász, and Schrijver~\cite{GLS1988} based on the ellipsoid method. The first
combinatorial strongly polynomial algorithms are due to Schrijver~\cite{Schrijver2000} and
Iwata, Fleischer, and Fujishige~\cite{IwataFleischerFujishige2001}; see
McCormick~\cite{McCormick2005} for a survey. The combinatorial algorithms of Schrijver and of
Iwata, Fleischer, and Fujishige are both based on a maximum flow-style algorithmic framework
of Cunningham that earlier led to the first pseudopolynomial algorithm for submodular function
minimization~\cite{BixbyCunninghamTopkis1985,Cunningham1984,Cunningham1985}.

Schlöter and Skutella~\cite{SchloeterSkutella2017,Schloeter2018} observe that checking
feasibility of a transshipment over time problem, as described in the proof of
Theorem~\ref{thm:sfm}, can directly produce a feasible solution if submodular function
minimization is done with an algorithm using Cunningham’s framework.

\begin{theorem}[Schlöter, Skutella~\cite{SchloeterSkutella2017,Schloeter2018}]
\label{thm:SchloeterSkutella}
A feasible solution to the transshipment over time problem can be found via one submodular
function minimization with an algorithm using Cunningham’s framework, and~$k$ calls of Hoppe
and Tardos' lexicographically maximum flow over time algorithm (Algorithm~\ref{alg:lexmax})
\end{theorem}

\begin{proof}
Cunningham's framework builds on the following strong duality result of
Edmonds~\cite{Edmonds1970}: For a submodular function~$g:2^{\terminals}\to\R$
\[
\min_{X\subseteq\terminals}g(X)=\max\bigl\{z^-(\terminals)\mid z\in B(g)\bigr\},
\]
where~$z^-(\terminals)$ is the sum of all negative entries of vector~$z\in
B(g)\subset\R^{\terminals}$. Cunningham's framework finds a subset~$X\subseteq\terminals$
minimizing~$g(X)$ together with an optimal dual solution~$z\in B(g)$ where~$z$ is represented
by a convex combination of~$k$ vertices of~$B(g)$; notice that~$k$ vertices suffice by
Carathéodory's theorem since the dimension of base polytope~$B(g)$ is at most~$k-1$ due to
equality constraint~$z(\terminals)=g(\terminals)$.

If, in our case of the submodular function~$g$ with~$g(X):=\oo(X)-b(X)$, the minimum function
value is~$0$, an optimal dual solution~$z\in B(g)$ has no negative entries. Therefore,
since~$z(\terminals)=g(\terminals)=o(\terminals)-b(\terminals)=0$, the only optimal dual
solution is the zero vector. In this case, a submodular function minimization algorithm based
on Cunningham's framework finds a representation of the zero vector as a convex combination of
vertices of~$B(g)$. Here, each vertex is given by a linear order on the ground set~$S$.

Moreover, it is not difficult to see that the base polytope~$B(g)$ of submodular function~$g$
with~$g(X)=\oo(X)-b(X)$ is a translation of base polytope~$B(\oo)$ by vector~$b$, that
is,~$B(g)=B(\oo)-b$. Therefore, the representation of the zero vector as a convex combination
of~$k$ vertices of~$B(g)$ yields a representation of~$b$ as a convex combination of~$k$
vertices of~$B(\oo)$, using the same coefficients.

Thus, what remains to be done algorithmically is to compute lexicographically maximum flows
over time for the~$k$ linear orders of terminals~$\terminals$ corresponding to the~$k$
vertices of~$B(\oo)$. Their convex combination, using the same coefficients, is a feasible
solution to the transshipment over time problem.
\end{proof}

\begin{rem}
As we discussed at the very end of Section~\ref{sec:lex-max-flows}, the lexicographically
maximum flows over time found by Algorithm~\ref{alg:lexmax} only use integral flow rates if
all arc capacities are integral. Unfortunately, this favorable property is not inherited by
the transshipment over time in Theorem~\ref{thm:SchloeterSkutella} since the coefficients of
the convex combination found by the submodular function minimization algorithm will in
general be fractional.

The transshipment over time algorithm given by Hoppe and Tardos~\cite{HoppeTardos2000},
however, always finds an integral transshipment over time. They employ~$2k$ parametric
submodular function minimizations to construct a carefully tightened network for which one
particular lexicographically maximum flow over time exactly satisfies the supplies and
demands and is thus a feasible transshipment over time.
\end{rem}

\subsubsection*{Results on the quickest transshipment problem from the literature}

As the title of Hoppe and Tardos' paper~\cite{HoppeTardos2000} suggests, they are not merely
interested in the transshipment over time problem but actually solve the more difficult
\defword{quickest transshipment problem}. In this problem, the time horizon~$\timehorizon$ is
not given but only the network~$\net{}$, and the task is to determine the minimum time
horizon~$\timehorizon$ together with a feasible transshipment over time with that time
horizon. 

The \defword{quickest flow problem} is the special case of the quickest transshipment problem
with a single source and sink, or alternatively, with several sources and sinks without
supplies and demands but only a specified flow value that needs to be sent from the sources
to the sinks. Burkard, Dlaska, and Klinz~\cite{BurkardDlaskaKlinz1993} observe that the
quickest flow problem can be solved in strongly polynomial time by incorporating Ford and
Fulkerson's maximum flow over time algorithm (Algorithm~\ref{alg:FordFulkerson}) into
Megiddo’s parametric search framework~\cite{Megiddo1979}. Lin and
Jaillet~\cite{LinJaillet2015} present a cost-scaling algorithm that solves the problem in the
same running time that Goldberg and Tarjan's cost-scaling algorithm~\cite{GoldbergTarjan1990}
needs to find a static minimum cost flow. Refining Lin and Jaillet's approach, Saho and
Shigeno~\cite{SahoShigeno2017} achieve the currently best known algorithm with strongly
polynomial running time.

For the quickest transshipment problem, Hoppe and Tardos~\cite{HoppeTardos2000} show how to
find the minimum time horizon~$\timehorizon$ by solving a parametric submodular function
minimization problem using Megiddo's parametric search framework~\cite{Megiddo1979}.
Schlöter, Skutella, and Tran~\cite{SchloeterSkutTran2022} present a sophisticated extension
of the discrete Newton method (or Dinkelbach's algorithm~\cite{Dinkelbach1967}) for computing
the minimum time horizon which, in terms of its running time, beats the parametric search
approach by several orders of magnitude. For the special case of a single source node or a
single sink node, an even faster algorithm is given by Schlöter~\cite{Schloeter2018} and
Kamiyama~\cite{Kamiyama2019}, based on a simpler variant of the discrete Newton method. For
the case where each arc also has a cost coefficient (independent of its transit time),
it is not difficult to show that finding a quickest transshipment of minimum cost can be reduced to solving a quickest transshipment problem without costs on a carefully chosen subnetwork~\cite{Skutella2023}.

% End matter

\thankyou{%
The author is much obliged to Lizaveta Manzhulina, Tom McCormick, Britta Peis, Miriam
Schlöter, and Khai Van Tran for insightful discussions on the topic of this paper. The
presentation of the paper has improved thanks to a referee's instructive comments and
suggestions.

This work is supported by the Deutsche Forschungsgemeinschaft (DFG, German Research
Foundation) under Germany's Excellence Strategy --- The Berlin Mathematics Research Center
MATH+ (EXC-2046/1, project ID: 390685689).
}

%You can also use BibTeX with the amsplain style.
%\bibliographystyle{amsplain}
\providecommand{\bysame}{\leavevmode\hbox to3em{\hrulefill}\thinspace}
\providecommand{\MR}{\relax\ifhmode\unskip\space\fi MR }
\providecommand{\MRhref}[2]{%
  \href{http://www.ams.org/mathscinet-getitem?mr=#1}{#2}
}
\providecommand{\href}[2]{#2}

% \bigskip
% For preference, use the BibTeX style \texttt{amsplain} and send the editors
% your \texttt{.bib} file with exactly the references you use.

\myaddress

%% Editors to uncomment the two lines below if survey ends on an odd-numbered page
%\newpage
%\mbox{}

\end{document}